\documentclass[runningheads,review,screen,anonymous,svgnames,x11names]{llncs}
\usepackage[T1]{fontenc}
\usepackage{graphicx}
\usepackage[hidelinks]{hyperref}
\usepackage{color}

\usepackage{csquotes}
\usepackage{todonotes}
\usepackage{xcolor}
\usepackage{adjustbox}
\usepackage{mathtools}
\usepackage{xspace}
\usepackage{proof}
\usepackage{semantic}
\usepackage{listings}
\usepackage{stmaryrd}
\usepackage{wrapfig}
\usepackage{tikz}
\usepackage[nameinlink, noabbrev]{cleveref}
\usepackage{caption}
\usepackage{subcaption}

\usepackage{amsmath}
\usepackage{amssymb}

\definecolor{codegreen}{rgb}{0,0.6,0}
\definecolor{codegray}{rgb}{0.5,0.5,0.5}
\definecolor{codepurple}{rgb}{0.58,0,0.82}
\definecolor{backcolour}{rgb}{0.95,0.95,0.92}

\lstdefinestyle{mystyle}{
	commentstyle=\color{codegreen},
	keywordstyle=\color{magenta},
	numberstyle=\tiny\color{codegray},
	stringstyle=\color{codepurple},
	basicstyle=\ttfamily\footnotesize,
	breakatwhitespace=false,         
	breaklines=true,                 
	captionpos=b,                    
	keepspaces=true,                 
	numbers=left,                    
	numbersep=5pt,                  
	showspaces=false,                
	showstringspaces=false,
	showtabs=false,                  
	tabsize=2
}
\DeclareMathAlphabet{\mathpzc}{OT1}{pzc}{m}{it}

\newcommand{\sfsymbol}[1]{\textsf{\upshape {#1}}}
\newcommand{\ttsymbol}[1]{\texttt{\upshape {#1}}}

\newcommand{\Varso}{\sfsymbol{Vars}_o}
\newcommand{\Varsu}{\sfsymbol{Vars}_u}

\newcommand{\mydot}{\text{{\Large\textbf{.}}~}}

\newcommand{\morespace}[1]{~{}#1{}~}
\newcommand{\qmorespace}[1]{\quad{}#1{}\quad}

\newcommand{\eeq}{~{}={}~}

\newcommand{\qeq}{\quad{}={}\quad}

\newcommand{\GG}{~{}>{}~}
\newcommand{\ggeq}{~{}\geq{}~}

\newcommand{\iin}{~{}\in{}~}
\newcommand{\nnotin}{~{}\not\in{}~}

\newcommand{\tto}{\morespace{\to}}

\newcommand{\ttriangleright}{\morespace{\triangleright}}
\newcommand{\qtriangleright}{\qmorespace{\triangleright}}

\newcommand{\wpsymbol}{\sfsymbol{wp}}

\renewcommand{\wp}[2]{\wpsymbol \left\llbracket {#1} \right\rrbracket \left( {#2} \right)}
\newcommand{\wpC}[1]{\wpsymbol \, \left\llbracket {#1} \right\rrbracket}

\newcommand{\SKIP}{\ttsymbol{skip}}

\newcommand{\AssignSymbol}{\coloneqq}
\newcommand{\ASSIGN}[2]{\ensuremath{#1 \AssignSymbol #2}}
\newcommand{\CHOOSE}[2]{\ASSIGN{#1}{\texttt{sample}\left(#2\right)}}
\newcommand{\BMODEL}[2]{\ensuremath{\mathit{Pr}_{#1}\left(#2\right)}}
\newcommand{\SMODEL}[2]{\ensuremath{#1\left(#2\right)}}

\newcommand{\COMPOSE}[2]{\ensuremath{{#1}{\,\fatsemi}~ {#2}}}

\newcommand{\IFSYMBOL}{\ensuremath{\textnormal{\texttt{if}}}}

\newcommand{\ELSESYMBOL}{\ensuremath{\textnormal{\texttt{else}}}}

\newcommand{\ITE}[3]{\ensuremath{\IFSYMBOL\,\left(\, {#1} \,\right)\,\left\{\, {#2} \,\right\}\,\ELSESYMBOL\,\left\{\, {#3} \,\right\}}}
\newcommand{\INFER}[3]{\ensuremath{\texttt{infer}\left(\, {#1} \,\right)\left\{\, {#2} \,\right\}\,\ELSESYMBOL\,\left\{\, {#3} \,\right\}}}

\newcommand{\observesymbol}{\textnormal{\texttt{observe}}}
\newcommand{\OBSERVE}[1]{\ensuremath{\observesymbol \,(#1)}}
\newcommand{\OBSERVEASSIGN}[2]{\ensuremath{#1 \coloneqq \observesymbol \,\left(#2\right)}}

\newcommand{\WHILESYMBOL}{\ensuremath{\textnormal{\texttt{while}}}}

\newcommand{\WHILEDO}[2]{\ensuremath{\WHILESYMBOL \, \left(\, {#1} \,\right)\left\{\, {#2} \,\right\}}}

\newcommand{\BOUNDEDWHILE}[3]{\ensuremath{\WHILESYMBOL^{{#1}}\:(#2)\:\{#3\}}}

\newcommand{\PBLIMP}{\sfsymbol{\mbox{pBLIMP}}\xspace}   

\newcommand{\Nats}{\ensuremath{\mathbb{N}}\xspace}

\newcommand{\PosRealsInf}{\mathbb{R}_{\geq 0}^\infty}

\newcommand{\iverson}[1]{\left[ {#1} \right]}

\newcommand{\statesubst}[2]{\left[ {#1} \mapsto {#2}\right]}

\newcommand{\monus}{
	\mathbin{
		\vphantom{+}
		\text{
			\mathsurround=0pt 
			\ooalign{
				\noalign{\kern-.5ex}
				\hidewidth$\smash{\cdot}$\hidewidth\cr 
				\noalign{\kern.5ex}
				$-$\cr 
			}%
		}%
	}%
}

\newcounter{blacktrianglelefteq}

\newcounter{leftsliceeq}

\newcommand{\nneq}{~{}\neq{}~}

\newcommand{\pplus}{~{}+{}~}

\newcommand{\qmid}{\;\;|{}\;\;}

\definecolor{webgreen}{rgb}{0,.5,0}

\colorlet{myblue}{DodgerBlue3}
\newcommand{\blue}[1]{\textcolor{myblue}{#1}}
\newcommand{\orange}[1]{\textcolor{DarkOrange3}{#1}}

\definecolor{mydarkorange}{RGB}{220,80,0}
\colorlet{cemphcolor}{mydarkorange}

\newcounter{computationarrowsone}
\newcounter{computationarrowstwo}

\newcounter{sarrow}

\newcommand{\lfp}{\ensuremath{\textnormal{\sfsymbol{lfp}}~}}

\makeatletter
\newlength{\negph@wd}
\DeclareRobustCommand{\negphantom}[1]{%
  \ifmmode
    \mathpalette\negph@math{#1}%
  \else
    \negph@do{#1}%
  \fi
}
\newcommand{\negph@math}[2]{\negph@do{$\m@th#1#2$}}
\newcommand{\negph@do}[1]{%
  \settowidth{\negph@wd}{#1}%
  \hspace*{-\negph@wd}%
}
\makeatother

\newcommand{\lfpX}[1]{\ensuremath{\text{lfp}\ X . #1}}

\newcommand{\diverge}{\ensuremath{{\textnormal{\textsf{diverge}}}}}

\newcommand{\listtraverse}[1]{C}

\newboolean{extended}
\setboolean{extended}{true}

\newcommand{\extendedVersion}{TODO}

\begin{document}

	\title{Programming and Reasoning in Partially Observable Probabilistic Environments}
	\titlerunning{Probabilistic Belief Programming}
	%
	\author{Tobias Gürtler\inst{1} \and
		Benjamin Lucien Kaminski\inst{1,2}\orcidID{0000-0001-5185-2324}}
	\authorrunning{T. Gürtler and B. Kaminski}
	
	\institute{Saarland University Saarland Informatics Campus, Saarbrücken, Germany\\
		\email{\{guertler, kaminski\}@cs.uni-saarland.de} \and
		University College London, London, UK}
	\maketitle              
	\begin{abstract}
		Probabilistic partial observability is a phenomenon occuring when computer systems are deployed in environments that behave probabilistically and whose exact state cannot be fully observed.
		In this work, we lay the theoretical groundwork for a probabilistic belief programming language -- \PBLIMP{} -- which maintains a probability distribution over the possible environment states, called a belief state.
		\PBLIMP has language features to symbolically model the behavior of and interaction with the partially observable environment and to condition the belief state based on explicit observations.
		In particular, \PBLIMP programs can perform state estimation and base their decisions (i.e. the control flow) on the likelihood that certain conditions hold in the current state.
		Furthermore, \PBLIMP features unbounded loops, which sets it apart from many other probabilistic programming languages.
		For reasoning about  \PBLIMP programs and the situations they model, we present a weakest-precondition-style calculus~($\wpsymbol$) that is capable of reasoning about \emph{unbounded} loops.
		Soundness of our $\wpsymbol$ calculus is proven with respect to an operational semantics.
		We further demonstrate how our $\wpsymbol$ calculus reasons about (unbounded) loops with loop invariants.
		
	\end{abstract}


\section{Introduction}
Probabilistic partial observability is a phenomenon which occurs when computer systems are deployed in an environment that behaves probabilistically and whose exact state cannot be fully observed~\cite{smallwood1973optimal}. 
Partial observability arises naturally when computer systems interact with the real world:
For example, spacecrafts have to estimate their current position and speed based on inaccurate or noisy measurements \cite{burkhart1996adaptive}, a problem which also occurs in autonomous robot navigation~\cite{simmons1995probabilistic} or target tracking~\cite{bilik2010maneuvering}.
Beyond cyber-physical systems, partial observability also affects fields like medicine, where it is unobservable whether a patient is truly infected with a disease or not.
Instead, one has to rely on tests which may randomly give wrong results.

All in all, we are dealing with scenarios with probabilistic partial observability:
Not only is the true state of the system influenced by randomness, but moreover can the true state not be fully observed.
A common approach to partial observability in probabilistic domains is to model such systems by a \emph{partially observable Markov decision process} (POMDP)~\cite{cassandra1998survey}, but these approaches require the state space of the system to be modeled explicitly.

Probabilistic programming languages (PPLs) are well suited to model partial observability \cite{agentmodels} compactly, as they can usually reason about conditional probabilities. 
PPLs such as Anglican \cite{tolpin2016design}, Dice \cite{holtzen2020scaling}, Psi \cite{gehr2020lpsi}, WebPPL \cite{dippl} or Stan \cite{carpenter2017stan} -- amongst many others -- offer a symbolic way to model partially observable environments, and they offer a wide variety of inference algorithms to reason about the state of the unobservable environment.
However, modern PPLs largely lack tools to soundly reason about the behavior of unbounded loops, and the exact analysis of probabilistic programs with unbounded loops remains an open problem \cite{junges2024scalable}.
To address this, we build upon the (non-probabilistic) work of Atkinson \& Carbin~\cite{atkinson2020programming} and present the probabilistic belief programming language \PBLIMP which allows us to model controllers for partially observable environments.
Additionally, we present a novel predicate transformer style $\wpsymbol$-calculus, which can symbolically reason about \PBLIMP programs featuring unbounded loops.

\noindent
\PBLIMP can be categorized as a model PPL with both nested inference and unbounded loops for discrete probabilistic environments.
At runtime, the program dynamically keeps track of a probabilistic belief state, i.e. a \emph{probability distribution} over true states. 
In order to model partially observable probabilistic environments, \PBLIMP offers several dedicated language features.

In particular, \PBLIMP features sampling from discrete probability distributions  to model the behavior of the unobservable environment.
\PBLIMP also includes a statement to make sound inferences about the state of the environment, and to use the result of those inferences to guide the control flow of the belief program.
Finally, under partial observability, it is still possible to obtain partial information about the true state of the environment, and so \PBLIMP includes the ability to explicitly model partial observations of the unobservable environment.
Making such observations has two effects:
For one, the execution randomly branches into a successor belief state for each possible observation, and secondly each successor belief state is conditioned with the corresponding observation.
This branching behavior differentiates \PBLIMP from other PPLs, as probabilistic programs usually reason about a single, fixed observation.
In contrast, \PBLIMP allows developers to verify the behavior of a program across all possible sequences of observations which the program may encounter.

To reason about \PBLIMP programs, we develop our main contribution:
A partial-observability-tailored weakest preexpectation ($\wpsymbol$) calculus \`{a} la McIver and Morgan \cite{DBLP:series/mcs/McIverM05}, allowing us to reason about the expected value of some quantity at termination. 
This weakest preexpectation calculus offers a novel way to symbolically reason about probabilistic programs with nested inference, and furthermore enables reasoning using \textit{loop invariants}.
To our knowledge our weakest preexpectation calculus is the first to handle nested inference.

\subsubsection{Contributions.}
Our work lays the theoretical groundwork to lift belief programming
\cite{atkinson2020programming} to probabilistic settings, and enables further research on the symbolic analysis of partial observability. 
We summarize our contributions as follows:
\begin{itemize}
	\item \textit{Probabilistic Belief Programming:} We present \PBLIMP, a model probabilistic programming language featuring nested inference and unbounded loops.
	
	\item \textit{Weakest Preexpectation Calculus:} We present a predicate-transformer-style calculus ($\wpsymbol$) to reason about the correctness of \PBLIMP programs. 
	
	\item \textit{Effective $\wpsymbol$ Calculations:} We introduce an expressive grammar for the properties that the above-mentioned $\wpsymbol$ calculus uses.
	Under certain reasonable assumptions, this grammar allows for an effective calculation of $\wpsymbol$.
\end{itemize}
	

\section{An Overview of \PBLIMP}\label{ExampleSection}
Consider the (ongoing) treatment of a patient suffering from the fictional disease \emph{exemplitis}.
Due to prior data, we know that the patient is sick with a probability of $90\%$, and we further know that the medicine we prescribe has a $25\%$ chance of curing the disease. 
Our goal is to prescribe medicine until the patient is cured.

Unfortunately, the true status of having exemplitis is unobservable.
Instead, we have to rely on a test which returns correct results with $95\%$ probability.
This imperfect test naturally introduces probabilistic partial observability to our problem, as the test only reveals partial information about the disease. 
Our goal can hence only be to prescribe medicine until we are \enquote{reasonable certain} that the patient is cured, i.e.\ we continue treatment until the chance for the patient to be sick is below some parameter~\texttt{q}.
We assume that this is the only criterion to discharge the patient, and thus they could remain in medical care arbitrarily long.
\Cref{CaseStudyLoopy} models the scenario above as a \PBLIMP program.%

\begin{wrapfigure}[15]{l}{0.6\textwidth}%
	\vspace*{-1.36\intextsep}%
	\begin{lstlisting}
d = sample(0.9|1> + 0.1|0>);
inCare = true;
while(inCare){
	infer(p(d = 1) > q) {
		d = sample(0.75|d> + 0.25|0>);
		t = sample(0.95|d> + 0.05|1-d>);
		observe t;
	} else {
		inCare = false;
	}
}
\end{lstlisting}
	\caption{Probabilistic belief program with a loop modeling the treatment of a patient until they are likely cured from a disease.
	}
	\label{CaseStudyLoopy}
\end{wrapfigure}%
\paragraph{Sampling.}
The true state of having exemplitis is modeled by the variable \texttt{d}.
Line~1 uses a sample statement to model that the patient initially actually has exemplitis with a probability of $90\%$.
The value of \texttt{d} is not directly accessible to the program. 
Instead, the \PBLIMP program updates the current belief state to a probability distribution where the variable assignment $\{\texttt{d} \mapsto 1\}$ has probability $\frac{9}{10}$ and \mbox{$\{\texttt{d} \mapsto 0\}$} has probability $\frac{1}{10}$.
Notably, the program is still \emph{deterministically} in a \emph{single} belief state.
We then set a boolean flag \texttt{inCare} to represent that the patient is in medical care, and enter an unbounded loop which models their treatment.

\paragraph{Inference.}
In each loop iteration we decide whether the patient should continue treatment or be discharged.
This is modeled with an infer statement, which can be seen as a generalized if statement: 
Infer statements use the current belief state to determine the probability that a condition holds, and branch based on that value.
In \Cref{CaseStudyLoopy}, the infer statement in Line 4 infers the likelihood that \texttt{d = 1} holds, i.e. the probability that the patient is still sick.
If the inferred probability is higher than the parameter~\texttt{q}, then the patient continues treatment in Lines 5-7, otherwise we terminate the loop by setting \texttt{inCare} to \texttt{false} in Line~9 .

\paragraph{Observation.}
Next, Lines 5-7 model the treatment of the patient.
In Line 5, the patient is given medicine which has a $25\%$ chance of curing exemplitis and setting \texttt{d} to \texttt{0}.
Afterwards, Line 6 and 7 model an imperfect test:
Line~6 describes the test's behavior, namely it has a $95\%$ chance to return the true status of having exemplitis, and a $5\%$ chance of a wrong result.
By nature of the sample statement, the true value of neither \texttt{d} nor \texttt{t} is observable.
In Line 7, we now \emph{explicitly observe} the test result \texttt{t}, which has two effects:
First, the execution \emph{probabilistically branches} into two possible belief states, one per possible outcome of the test.
Second, in both belief states the value of \texttt{t} is now \emph{determined} and in each belief state the underlying probability that the patient is sick is automatically \emph{conditioned} on the observed test outcome.

This branching behavior is non-standard.
In most PPLs, observe statements would only condition the program behavior with a single test result, i.e. the developer has to choose if there will be a positive or a negative test.
In contrast, \PBLIMP branches and considers all possible test results and their likelihood to occur.
Hence, we can evaluate the effectiveness of our treatment strategy across all possible test results, and even all possible sequences of test results. 

\paragraph{Verification}
We can now ask an important question about the program in \Cref{CaseStudyLoopy}:
Across all possible sequences of test results, how likely are we to discharge a sick patient?
Put differently, what is the probability that the program terminates with \verb|d = 1|?
Many tools have been developed to answer such queries, but our example includes several features which complicate inference:

For one, the infer statement forces us to perform \textit{nested inference}, meaning we need to perform inference about programs which may themselves perform inferences. 
In practice, nested queries naturally arise from programs which perform probabilistic inference as a sub-routine, for example to track moving objects \cite{cheng2025inference,randone2024inference}, or to perform state estimation \cite{thrun2002probabilistic}.
Reasoning about programs with such nested queries in turn requires nested inference.
While tools for nested inference have been developed \cite{dippl,gehr2016psi,tolpin2016design}, it remains non-standard; e.g.\ Dice \cite{holtzen2020scaling} or SOGA \cite{randone2024inference} do not support nested queries. 
Additionally, our example included the \emph{parameterized} threshold \texttt{q}, requiring to perform \emph{symbolic} inference and the result of our analysis needs to be parameterized in \texttt{q}.
This leads us to PPLs like PSI \cite{gehr2020lpsi}, which implements exact symbolic inference methods and can answer our query. 
However, PSI does not support unbounded loops.
In fact, many, if not most PPLs do not support unbounded loops and instead opt to approximating them via bounded loops. 
This, however, yields no guarantees on the quality of the approximation \cite{torres2024iteration}.

To address this gap, we present a weakest preexpectation calculus $\wpsymbol$  \`{a} la McIver and Morgan \cite{DBLP:series/mcs/McIverM05}, which can perform exact symbolic inference on programs with loops.
The $\wpsymbol$-calculus enables reasoning about the expected value of some quantity at termination.
Furthermore, programs with unbounded loops can be analyzed using \emph{loop invariants}, which provide sound upper bounds on the behavior of \PBLIMP programs.
For the program from \Cref{CaseStudyLoopy}, we can formally verify that the probability to discharge a sick patient is always less than \texttt{q}.


\section{\PBLIMP{} --  Syntax and Semantics}

\subsection{Syntax of \PBLIMP}

\begin{figure}[t]
		\begin{align*}
			E_o &\qmorespace{\Coloneqq}\, x_o\in \mathit{\Varso}\,|\, E_o\oplus E_o\,|\, c \in \mathbb{N} & \oplus \in \{+, -, \cdot, /\}\\
			E_u &\qmorespace{\Coloneqq}\, x_u\in \mathit{\Varsu}\,|\, E_u\oplus E_u \,|\, E_o\\[.5em] 
			P_o &\qmorespace{\Coloneqq}  E_o \sim E_o  \,|\, P_o \land P_o \,|\, P_o \lor P_o \,|\, !P_o \,|\, b \in \{1, 0\} & \sim\, \in \{<,\leq, =, \neq, \geq, >\}\\
			P_u &\qmorespace{\Coloneqq} E_u\sim E_u \,|\, P_u \land P_u \,|\, P_u \lor P_u \,|\, !P_u \,|\, P_o 
		\end{align*}
		\begin{align*}
			C \quad{}\Coloneqq{}\quad &\, \SKIP \qmid \COMPOSE{C}{C}  \qmid \ASSIGN{x_o}{E_o} \qmid \ITE{P_o}{C}{C} \qmid \WHILEDO{P_o}{C}\\
			&\qmid \OBSERVEASSIGN{y_o}{x_u} \qmid \INFER{p(P_u) \in i}{C}{C}\\
			& \qmid \CHOOSE{x_u}{f}~, \quad \textnormal{where }f\colon \Sigma \rightarrow \mathcal{D}(\mathbb{N})
		\end{align*}
	\caption{Formal definition of the model language \PBLIMP.}\label{syntax_PBLIMP}
\end{figure}

\Cref{syntax_PBLIMP} describes the syntax of \PBLIMP.
Notably, we consistently delineate between \textit{observable} and \textit{unobservable} constructs, as seen e.g.\ with variables: 
Variables are partitioned into observable variables $x_o \in \Varso$ and unobservable variables $x_u \in \Varsu$. 
Observable variables are like variables in most programming languages: their values are freely accessible.
Values of unobservable variables are not freely accessible, and not all program statements may use them.

\subsubsection{Expressions and Propositions.} 

Similarly to variables, we use $E_o$ to refer to observable expressions and $E_u$ to refer to unobservable expressions. 
Observable expressions are either a numeric constant $c\in \mathbb{N}$, an \emph{observable} variable~$x_o$, or are recursively constructed using the standard binary operators $+,-,\cdot,/$ where subtraction truncates at~0. 
Notably, \emph{no un}observable variable may occur in an observable expression.
Unobservable expressions are constructed in the same way, but they \emph{may} contain \emph{un}observable variables $x_u$.

Propositions are Boolean expressions comparing (un)observable expressions.
Analogous to expressions, we distinguish observable propositions $P_o$ and unobservable propositions $P_u$, see \Cref{syntax_PBLIMP}.

\subsubsection{\PBLIMP Statements.}

The syntax of \PBLIMP statements is defined in \Cref{syntax_PBLIMP}, and we use~$C$ as a metavariable to range over \PBLIMP statements. 
\PBLIMP is based on the imperative model programming language \sfsymbol{IMP} \cite{winskel1993formal}, and hence contains all standard imperative programming concepts such as skipping, assignments, loops, etc. 
Notably, all statements which \PBLIMP inherits from \sfsymbol{IMP} are restricted to observable expressions or propositions. 
This syntactically guarantees that all standard statements 
behave in a completely standard way (apart from being restricted to positive numbers).
In addition to \sfsymbol{IMP}, \PBLIMP offers three additional, non-standard, statements over which we go in the following.

\paragraph{Sampling.} 

A \emph{sampling} statement $\CHOOSE{x_u}{f}$ assigns a random value to an \emph{un}observable variable $x_u$ using a function $f$ of type $\Sigma \rightarrow \mathcal{D}(\mathbb{N})$ which maps a variable assignment $\sigma \in \Sigma$ to a discrete \emph{probability distribution} over possible values in~$\mathbb{N}$. 
Here, $f(\sigma)(c)$ describes the probability to sample the value $c \in \mathbb{N}$, assuming that the current state aligns with the variable assignment $\sigma$.
We use the notation $f(\sigma, c) = f(\sigma)(c)$ for readability.

\paragraph{Observing.} 

$\OBSERVEASSIGN{y_o}{x_u}$ allows to \emph{explicitly observe (or measure) the value of an unobservable variable $x_u$}, thus deliberately lifting the veil of partial observability from the value of $x_u$.
For technical reasons, $x_u$ remains an \emph{un}observable variable even after the observe statement, as otherwise $x_u$ would be typed ambiguously after executing \verb|if(P) {observe x} else {skip}|. 
To anyway have access to the observed value of $x_u$, that value is assigned to an \emph{observable} variable $y_o$.
We use \OBSERVE{x_u} as syntactic sugar for \OBSERVEASSIGN{y_o}{x_u}, where $y_o$ is never used in the program (see e.g.\ Line 7 in \Cref{CaseStudyLoopy}).

\paragraph{Inferring.} 
The infer statement $\INFER{p(P_u) \in i}{C_1}{C_2}$ can be understood as a generalized if-statement. 
Rather than being restricted to observable propositions as branching conditions, infer statements can branch on an unobservable proposition $P_u$. 
The infer statement guides the control flow by determining the probability that $P_u$ holds, and executes either $C_1$ or $C_2$, depending on whether or not that probability falls into the interval $i$ or not.

\begin{remark}[Restrictions on Unobserved Variables]
	We would like to justify the restrictions on unobservable variables with an example:
	Anna and Bob are playing a game. 
	Anna tosses a coin so that Bob cannot see the result, and Bob has to guess the outcome. 
	If Bob guesses correctly, he wins the game, otherwise Anna wins. 
	If Anna is not lying and the coin is fair, then every possible guessing strategy of Bob should have a win-rate of $50\%$. 
	Assume now that we allow unobservable propositions as conditions of if-statements. 
	Then Bob could model his strategy and the outcome of the game with the following program:
\begin{lstlisting}
coin = sample(0.5|Heads> + 0.5|Tails>);
if (coin = Heads) {guess = Heads;} else {guess = Tails;}
actual = observe coin;
\end{lstlisting}%
				Here, \texttt{0.5|Heads> + 0.5|Tails>} is a distribution which assigns the value \texttt{Heads} to \texttt{coin} with a probability of $\frac{1}{2}$, and similarly for \texttt{Tails}.
				Bob makes his guess in Line~2 using an if-statement, which uses \texttt{coin} in its condition to ignore the partial observability which was introduced by the sample statement.
				Afterwards, we observe the actual outcome of the coin, and we find that the program will always terminate in a state where \texttt{actual} and \texttt{guess} are equal, thus giving Bob an impossible win-rate of $100\%$. 
				\hfill$\triangleleft$
			\end{remark}

			\subsection{Operational Semantics of \PBLIMP}

			Before defining operational semantics of \PBLIMP, we first need to define the notions of variable assignments and belief states. 
			
			\paragraph{Variable Assignments.}
			
			A variable assignment $\sigma \in \Sigma$ assigns a value $c \in \mathbb{N}$ to every (observable and unobservable) variable $x$, written as $\sigma(x) = c$. 
			We denote by $\sigma  (E)$ the evaluation of an expression $E$ under $\sigma$ (to a number in $\Nats$), and by~$\sigma  (P)$ the evaluation of a proposition $P$ under $\sigma$ (to a Boolean in $\{0,\, 1\}$). 
			
			\paragraph{Belief States.}
			
			A belief state $\beta \in \mathcal{D}(\Sigma)$ is a probability distribution over variable assignments~\mbox{$\sigma \in \Sigma$} and we denote by $\beta(\sigma)$ the probability of $\sigma$ under $\beta$. 
			Furthermore, we will use bra-ket notation $\beta = p_1{| \sigma_1 \rangle} + {\ldots} +  p_n {| \sigma_n \rangle}$ to denote a belief state with $\beta(\sigma_i) = p_i$. 
			Belief states capture the inherent uncertainty in \PBLIMP programs which have no direct access to the values of the unobservable variables, i.e.\  $\beta(\sigma)$ represents the likelihood that $\sigma$ is the true variable assignment.
			
			We say that $\sigma$ is \emph{possible} within $\beta$ if $\beta(\sigma) > 0$. 
			We call  $\beta$ \emph{consistent} if all possible variable assignments within $\beta$ agree on the values of the \emph{observable} variables. 
			For the remainder of this paper, we will only consider consistent belief states.
			Notably, all operations on belief states we present will preserve consistency.
			We can \emph{evaluate a proposition $P$ under a belief state $\beta$} as follows:%
			\begin{align*}
				\BMODEL{\beta}{P} \morespace{\coloneqq} \sum_{\sigma\in \Sigma}  \beta (\sigma) \cdot \SMODEL{\sigma}{P}	\tag{$\star$}\label{eval-page}
			\end{align*}%
			$\BMODEL{\beta}{P}$ yields a value in $[0,1]$ and can be understood as the likelihood that $P$ holds in the true state if $\beta$ is our current belief state. 
			If $\beta$ is consistent and $P$ is an observable proposition, then either $\BMODEL{\beta}{P} = 0$ or $\BMODEL{\beta}{P} = 1$ holds.

			\paragraph{Configurations.}
			
			A program configuration is a quadruple $(C,\beta, \sigma, p)$, where:
			\begin{itemize}
				\item 
				$C$ is either a \PBLIMP statement that remains to be executed or the special symbol $\top$ that indicates termination.
				\item 
				$\beta$ is the current belief state.
				\item 
				$\sigma$ is the current (true) variable assignment, which captures the true state of both the observable and unobservable  variables. 
				\item 
				$p$ stores the probability of the execution path up to this configuration.
			\end{itemize}
			Here, $C$ and $\beta$ capture the state of the program, as well as its current knowledge about the true state of the (unobservable) environment.
			However, the program directly interacts with the environment during observations, and hence $\sigma$ is included in the program configuration to model the unobservable environment.
			We now provide a small-step operational semantics that steps through the program by applying the transition rules in \Cref{Transitions}. 
			Let us go over these rules.
			\begin{figure}[t]
				\begin{center}
					\begin{adjustbox}{max width=\textwidth}
						{\renewcommand{\arraystretch}{3.16}
							\begin{tabular}{c c}

								\inference{
									\beta\statesubst{x_o}{E_o}(\sigma') \coloneqq \sum\limits_{
										\rho\statesubst{x_o}{E_o} \eeq \sigma'} \beta(\rho)}
								{ (\ASSIGN{x_o}{E_o},\, \beta,\, \sigma,\, p) \qtriangleright (\top,\,\beta\statesubst{x_o}{E_o},\, \sigma\statesubst{x_o}{E_o},\, p) }
								
								&
								
								\inference{}{  (\SKIP,\,\beta,\, \sigma,\, p) \qtriangleright  (\top,\, \beta,\, \sigma,\, p) }
								\\

								\multicolumn{2}{c}{
									\inference{f(\sigma,\, c) \eeq q \qquad\qquad\qquad q \GG 0}
									{ (\CHOOSE{x_u}{f},\,\beta,\, \sigma,\, p) \qtriangleright (\top,\,\beta\statesubst{x_u}{f},\, \sigma\statesubst{x_u}{c},\, p \cdot q) }
								}\\
								
								\multicolumn{2}{c}{\inference{\sigma(x_u) \eeq c}{  (\OBSERVEASSIGN{y_o}{x_u},\, \beta,\, \sigma,\, p) \qtriangleright  (\top,\, \beta|_{x_u=c}\statesubst{y_o}{c},\, \sigma\statesubst{y_o}{c},\, p) }}\\


								\multicolumn{2}{c}{
									\inference{\BMODEL{\beta}{P_o} \eeq 1}
									{ (\ITE{P_o}{C_1}{C_2},\, \beta,\, \sigma,\, p) \qtriangleright (C_1,\,\beta,\, \sigma,\, p) }
								}\\
								\multicolumn{2}{c}{
									\inference{\BMODEL{\beta}{P_o} \eeq 0}
									{ (\ITE{P_o}{C_1}{C_2},\, \beta,\, \sigma,\, p) \qtriangleright (C_2,\,\beta,\, \sigma,\, p) }
								}\\
								
								\multicolumn{2}{c}{\inference{}
									{ (\WHILEDO{P_o}{C},\,\beta,\, \sigma,\, p) \qtriangleright (\ITE{P_o}{\COMPOSE{C}{\WHILEDO{P_o}{C}}}{\SKIP},\, \beta,\, \sigma,\, p) }} \\
								
								\inference{\BMODEL{\beta}{P_u}\iin i}
								{ (\INFER{P_u \in i}{C_1}{C_2},\, \beta,\, \sigma,\, p) \qtriangleright (C_1,\,\beta,\, \sigma,\, p) }&
								
								\inference{(C_1,\, \beta,\, \sigma,\, 1) \ttriangleright (C_1',\, \beta',\, \sigma',\, q) \qquad C_1' \nneq \top}
								{ (\COMPOSE{C_1}{C_2},\, \beta,\, \sigma,\, p) \qtriangleright (\COMPOSE{C_1'}{C_2},\, \beta',\, \sigma',\, p \cdot q) } 
								
								\\
								
								\inference{\BMODEL{\beta}{P_u} \nnotin i}
								{ (\INFER{P_u \in i}{C_1}{C_2},\, \beta,\, \sigma,\, p) \qtriangleright (C_2,\,\beta,\, \sigma,\, p) } 
								&
								
								\inference{(C_1,\, \beta,\, \sigma,\, 1) \ttriangleright (\top,\, \beta',\, \sigma',\, q)}
								{ (\COMPOSE{C_1}{C_2},\, \beta,\, \sigma,\, p) \qtriangleright (C_2,\, \beta',\, \sigma',\, p\cdot q) }
								
						\end{tabular}}
					\end{adjustbox}
				\end{center}
				
				\caption{Operational transition rules for configurations of the form $(C,\,\beta,\, \sigma,\, p)$. 
				$\beta\statesubst{x_u}{f}$ and $\beta|_{x_u=c}$ are defined in \ref{dagger-page} and \ref{ddagger-page} respectively.
}
				\label{Transitions}
			\end{figure}%

			\paragraph{Assignments.} 
			
			$\ASSIGN{x}{E_o}$ deterministically modifies the value of the observable variable $x_o$ to the observable expression $E_o$. 
			We thus need to update the value of $x_o$ in the true variable assignment~$\sigma$ and in all possible variable assignments within $\beta$. 
			The updated variable assignment is given by:
			\begin{align*}
				\sigma\statesubst{x}{E}(y)\morespace{\coloneqq}\begin{cases} 
					\sigma(y) &,\,x\neq y\\
					\sigma(E) &,\,x = y
				\end{cases}
			\end{align*} 
			$\beta\statesubst{x_o}{E_o}$ is obtained from $\beta$ by applying $\statesubst{x_o}{E_o}$ to all possible variable assignments within~$\beta$ and possibly summing the probability masses should two or more formerly different variable assignments coincide after applying~$\statesubst{x_o}{E_o}$.
			
			\paragraph{Sampling.}
			The sample statement models the behavior of the environment by assigning a new value to a variable $x_u$ via sampling from the distribution $f(\sigma)$. 
			At this point, the configuration has multiple possible successors, each with a different true variable assignment. 
			To account for the program's uncertainty about the true variable assignment, the belief state splits each possible variable assignment $\rho$ within $\beta$ in accordance with the distribution $f(\rho)$ to create a new belief state. 
			Formally, the resulting belief state is defined as%
			\begin{align*}
				\beta\statesubst{x}{f} 
				\morespace{\coloneqq}
				\lambda \sigma\mydot \quad \sum\limits_{\mathclap{\rho \in \Sigma, \rho[x\mapsto \sigma(x)] = \sigma}}~ \beta(\rho) \cdot f(\rho, \sigma(x))~. \tag{$\dagger$}\label{dagger-page}
			\end{align*}%
			Here, the likelihood of a variable assignment $\sigma$ \textit{after} sampling, is described by considering all variable assignments $\rho$ \textit{before} sampling, so that $\rho$ aligns with $\sigma$ if $x$ is set to $\sigma(x)$. 
			The updated probability for $\sigma$ can then be calculated as the  sum over all such variable assignments $\rho$, where we consider the probability to \enquote{start} in $\rho$, i.e. $\beta(\rho)$, and to then sample the value $\sigma(x)$, which is $f(\rho, \sigma(x))$.
			The premise $q > 0$ prevents division by $0$ in the rule for the observe statement.
			
			\paragraph{Observing.} $\OBSERVEASSIGN{y_o}{x_u}$ has two effects: Determinizing the value of $x_u$ within the belief state $\beta$, and assigning the value of $x_u$ to the variable $y_o$. 
			The latter effect is analogous to the assignment, and we will thus not elaborate further on the effect it has on the belief state and the true variable assignment.
			
			The true variable assignment $\sigma$ models the state of the environment which is hidden behind partial observability, and thus $\sigma$ fixes the true value of $x_u$. 
			The observe statement passes the true value of $x_u$ to the program, and uses it to condition the belief state:
			As we now know the true value of $x_u$ from $\sigma$, some of the variable assignment within $\beta$ are potentially no longer possible, simply because they map $x_u$ to the wrong value. 
			We incorporate this information into~$\beta$ by filtering out variable assignments which do not match our observation, and then re-normalizing the distribution. 
			Formally, this is described as:%
			\begin{align*}
				\beta|_{x=c} \coloneqq \lambda \sigma\mydot \begin{cases}
					\frac{\beta(\sigma) \cdot (\SMODEL{\sigma}{x=c})}{\BMODEL{\beta}{x=c}}~, &\textnormal{if } \BMODEL{\beta}{x=c} > 0\\
					0~,& \text{otherwise.}	
				\end{cases}
				\tag{$\ddagger$}
				\label{ddagger-page}
			\end{align*}%
			For example, observing $x$ on belief state $\tfrac{1}{3} | \blue{0}\,0 \rangle + \tfrac{1}{3} | \orange{1}\,0 \rangle + \tfrac{1}{6} | \blue{0}\,1 \rangle + \tfrac{1}{6} | \orange{1}\,1 \rangle$ will branch with probability $\tfrac{1}{3} + \tfrac{1}{6} = \tfrac{1}{2}$ into $\tfrac{2}{3} | \blue{0}\,0 \rangle + \tfrac{1}{3} | \blue{0}\,1\rangle$ and with probability $\tfrac{1}{3} + \tfrac{1}{6} = \tfrac{1}{2}$ into $\tfrac{2}{3} | \orange{1}\,0 \rangle + \tfrac{1}{3} | \orange{1}\,1\rangle$.
			This branching is not directly represented in the operational semantics above, but it is the result of the sample statements which set up the believe state \textit{prior} to the observation and branched out into multiple configurations with the same belief state.
			However, for the program, which only has access to the belief state, it appears as though the execution branches at the observe statement.

			\paragraph{Conditional Branching and Inferring.}
			Both the if- and the infer-statement pick a branch depending on their condition. 
			For the if-statement, the condition is an observable proposition $P_o$. 
			The configuration transitions to $C_1$ if $\BMODEL{\beta}{P_o} = 1$, otherwise the configuration transitions to $C_2$. 
			We always have either $\BMODEL{\beta}{P_o} = 1$ or $\BMODEL{\beta}{P_o} = 0$ for all observable propositions, as we only consider consistent belief states.
			For the infer-statement, the condition is an unobservable proposition $P_u$ which is evaluated under the belief state $\beta$ (cf. ($\star$) on~p.~\pageref{eval-page}). 
			The configuration transitions to $C_1$ if the value of $\BMODEL{\beta}{P_u}$ falls into the interval $i$, otherwise the configuration transitions to $C_2$. 
			\paragraph{While Loop, Composition and Skip.} The remaining statements follow their typical semantics. 
			The while loop is unrolled into if-statements until the loop guard is no longer satisfied. 
			The composed statement $\COMPOSE{C_1}{C_2}$ first executes $C_1$ step-wise until termination, i.e.\ until $\top;C_2$ would be reached, and then behaves like~$C_2$. 
			The $\SKIP$ statement terminates with no effect. 
			
			\subsection{Execution and Termination}
			Starting in a configuration $(C,\,\beta,\, \sigma,\, p)$, we can exhaustively apply the transition rules to create a computation tree which includes all possible execution paths. 
			In this tree, the leaves are terminal configurations of the form $(\top,\,\beta',\, \sigma',\, p')$. 
			If the root configuration of the tree satisfies~\mbox{$p = 1$}, then $p'$ encodes the probability to reach that particular leaf. 
			However, note that the tree may contain multiple copies of this leaf, and thus $p'$ by itself is \emph{not} the probability to reach the  configuration~\mbox{$(\top,\,\beta',\, \sigma',\, p')$}.
			We can take the sum over all leaves of the tree to get a sub-distribution over the terminal configurations. 
			We consider sub-distributions, i.e.\ probability distributions which possibly add up to less than $1$, as there may be diverging paths in the computation tree. 
			We will write $[C]^{\beta, \sigma}$ for the sub-distribution created by the tree starting from the configuration $(C,\,\beta,\, \sigma,\, 1)$. 
			
			$[C]^{\beta, \sigma}$ can be used to determine the probability to reach a terminal configuration containing the belief state $\beta'$, written as $[C]^{\beta, \sigma} (\beta')$,  
			by considering all terminal configurations in which $\beta'$ occurs, and summing up the probabilities to reach those configurations. 
			Analogously, we can also determine the probability to terminate in configuration containing the belief state $\beta'$ and the variable assignment $\sigma'$, written as $[C]^{\beta, \sigma}(\beta', \sigma')$. 
			We can consider execution from an initial belief state as the expected outcome over all possible initial configurations with different true variable assignments $\sigma$. 
			Thus, we define the probability to reach a terminal state containing $\beta'$ starting from a belief state $\beta$ as:
			
			$$[C]^{\beta}(\beta') \coloneqq \sum_{\sigma\in \Sigma}\beta(\sigma) \cdot [C]^{\beta, \sigma} (\beta')$$

					\paragraph{Belief State Accuracy.} We can now ask how the belief state $\beta$ and the true variable assignment $\sigma$ relate. 
					In particular, if we terminate with a belief state~$\beta$, we expect the probability to also terminate with the true system state $\sigma$ to be $\beta(\sigma)$. 
					This ensures that we can make statements about the true variable assignment based only on the belief state, which is necessary as the programmer only has access to the latter. 
					We can formalize this property as follows:%
					\begin{theorem}[Belief Accuracy Preservation]\label{beliefAccuracyTheorem}
						Let $C$ be a program, $\beta_0$ be the initial belief state, $\beta$~be a belief state, and $\sigma$ be a variable assignment. 
						Then we have:
						$$[C]^{\beta_0} (\beta, \sigma) \eeq \beta(\sigma) \cdot [C]^{\beta_0} (\beta)$$
					\end{theorem}%
					\begin{proof}
						By structural induction on $C$. We refer to \ifthenelse{\boolean{extended}}{\Cref{belief_accuracy}}{the extended version~\extendedVersion}. \qed
					\end{proof}%
					\noindent%
					This theorem guarantees that the belief state accurately represents the unobservable environment. 
					Due to this, it is possible to drop the true variable assignment from the semantics altogether, and to define a semantics where configurations only maintain the belief state.
					Notably, program executions then branch at the observe statement, whereas the sample statement updates the belief state in a deterministic manner.
					For details on this we refer to \ifthenelse{\boolean{extended}}{\Cref{semantics_full}}{the extended version~\extendedVersion}.


\section{Weakest Preexpectations for \PBLIMP}
This section introduces a weakest preexpectation calculus for \PBLIMP programs. 
The calculus enables verification of 
quantitative properties on \PBLIMP programs, thus yielding hard guarantees on the execution of such programs.

\subsection{Weakest Preexpectation Calculus}

Usually, predicates are functions mapping to $1$ or $0$, depending on whether the predicate holds or not.
To account for the probabilistic uncertainty within \PBLIMP, we generalize to quantitative predicates which map belief states to any non-negative value (or infinity), i.e.\ we define a predicate $F$ to be a function of type $F\colon\mathcal{D}(\Sigma) \rightarrow \PosRealsInf$. 
Thanks to \Cref{beliefAccuracyTheorem}, these predicates still allow us to reason about the true variable assignment, e.g.\ the predicate $F(\beta) \coloneqq \BMODEL{\beta}{ x > 3}$ describes the probability that $x > 3$ holds in the true variable assignment.

Including infinity in the co-domain of predicates is a technical requirement for obtaining a \emph{complete lattice} of predicates.
We adhere to the convention $0 \cdot \infty = 0$ and moreover $\frac{a}{0} = 0$ for all $a$ including $\infty$.
We denote by $Pr(P)$ the predicate $\lambda\, \beta \mydot \BMODEL{\beta}{P}$, and we define predicate manipulations such as $F[x \mapsto E] \coloneqq \lambda\, \beta \mydot F (\beta\statesubst{x}{E})$. 
Similarly, we define $F\statesubst{x}{f}$ (cf.\ ($\dagger$) on~p.~\pageref{dagger-page}) and $F|_{x=c}$ (cf.\ ($\ddagger$) on~p.~\pageref{ddagger-page}).
To reason about predicates in the context of belief programs, we present a predicate transformer function $\wpsymbol$, which is designed in the style of McIver \& Morgan \cite{DBLP:series/mcs/McIverM05}. 
For our purpose, $\wpsymbol$ is of type%
\begin{align*}
	\wpsymbol\colon \quad \PBLIMP \tto \bigl(\mathcal{D}(\Sigma) \to \PosRealsInf\bigr) \tto \bigl(\mathcal{D}(\Sigma) \to \PosRealsInf\bigr)	~.
\end{align*}%
The rules which define $\wpsymbol$ can be found in \Cref{wp}. 
To our knowledge, Olmedo et al. \cite{olmedo2018conditioning} and Symczakk \& Katoen \cite{DBLP:conf/setss/SzymczakK19} (and in a broader sense also Nori et al. \cite{nori2014r2}) are the only works which present weakest preexpectation calculi that include conditioning statements.
Still, none of those works support infer-statements. 
\begin{figure}[t]
	\begin{center}
		\begin{adjustbox}{max width=\textwidth}
		\renewcommand{\arraystretch}{1.475}
		\begin{tabular}{l@{\quad}l}
			$C$ & $\wp{C}{F}$\\ \hline \hline
			\SKIP & $F$ \\
			$\ASSIGN{x}{E}$ & $F\statesubst{x}{E}$ \\
			$\ITE{P_o}{C_1}{C_2}$ & $Pr(P_o)  \cdot \wp{C_1}{F}\, + (1- Pr(P_o)) \cdot \wp{C_2}{F}$\\
			$\WHILEDO{P_o}{C}$ & $ \lfp X\mydot  (1 - Pr(P_o)) \cdot F \pplus Pr(P_o) \cdot \wp{C}{X}$ \\
			$\COMPOSE{C_1}{C_2}$ & $\wp{C_1}{\wp{C_2}{F}}$ \\
			$\CHOOSE{x}{f}$ & $ F\statesubst{x}{f}$ \\
			$\OBSERVEASSIGN{y}{x}$ & $\sum\limits_{c \in \mathbb{N}} Pr(x = c) \cdot F|_{x = c}\statesubst{y}{c}$ \\
			$\INFER{p(P) \in i}{C_1}{C_2}$ & $[Pr(P) \in i] \cdot \wp{C_1}{F}\, + [Pr(P) \notin i] \cdot \wp{C_2}{F}$
		\end{tabular}
	\end{adjustbox}
		\caption{Rules for the weakest preexpectation calculus $\wpsymbol$ for \PBLIMP.}
		\label{wp}
	\end{center}
\end{figure}%

The function $\wpsymbol$ determines the weakest preexpectation of a given belief program $C$ with respect to a predicate $F$. 
Intuitively, this can be understood as \enquote{predicting} the value of $F$ in the terminal state which is reached after executing the program $C$.
But, belief programs are probabilistic, meaning there may not be just a single terminal state to reach.
Hence, $\wpsymbol$ instead determines the \emph{expected value} of $F$ after executing $C$, hence also the term weakest pre\emph{expectation}.
Formally, $\wp{C}{F}\, (\beta)$ can be understood as the expected value of the predicate $F$ on all terminal states after executing the program $C$ starting in the belief state $\beta$.
We now consider the rules from \Cref{wp} with this intuition in mind.%

\paragraph{Assignment and Sampling.} 

As described above, $\wpsymbol$ determines the expected value of the predicate~$F$ after executing the program $C$. 
For the assignment, there is only a single possible terminal belief state, namely $\beta\statesubst{x}{E}$. 
We can thus determine the expected value of $F$ at termination by evaluating $F$ in the belief state $\beta\statesubst{x}{E}$, which amounts to the predicate $F\statesubst{x}{E}$. 
Similarly, the sample statement also only has a single reachable terminal belief state, namely $\beta\statesubst{x}{f}$, and thus its weakest preexpectation is $F\statesubst{x}{f}$.

\paragraph{Observing.} 
Consider the execution of an observe statement from the initial belief state $\beta$.
Notably, if we observe $x$ to have value $c$,  which occurs with probability~$\BMODEL{\beta}{x=c}$ by \Cref{beliefAccuracyTheorem}, then we always terminate with the belief state $\beta|_{x=c}\statesubst{y}{c}$. 
In line with the intuition for $\wpsymbol$, the weakest preexpectation thus determines the expected value of $F$ over all such observations:
When we observe $x$ to be $c$, which happens with probability $Pr(x=c)$, then we terminate with the belief state $\beta|_{x=c}\statesubst{y}{c}$, which will yield the \enquote{reward} $F|_{x=c}\statesubst{y}{c}$.

\paragraph{Conditionals and Inference.} 

Both the if- and the infer-statement execute either $C_1$ or $C_2$, depending on whether or not their condition holds. 
If the condition holds, then $C_1$ is executed and we can predict the expected value of $F$ with $\wp{C_1}{F}$, otherwise $C_2$ is executed and we can use $\wp{C_2}{F}$.
What remains is to distinguish the two cases:
For the if-statement, the condition $P_o$ is an observable proposition and $\BMODEL{\beta}{P_o}$ evaluates to either $0$ or $1$.
We can thus multiply with $Pr(P_o)$ to \enquote{select} the correct preexpectation.
For the infer-statement, the same principle applies, but we use the Iverson bracket $[Pr(P) \in i]$ to distinguish the two cases, where we define $[\varphi]$ as:%
\begin{align*}
	\iverson{\varphi}(\beta) \coloneqq \lambda\, \beta\mydot \begin{cases}
		1, & \text{if $\beta$ satisfies $\varphi$}\\
		0, & \text{otherwise}
	\end{cases}
\end{align*}

\paragraph{While Loop.} 

For the while-loop, $\wpsymbol$ is defined as the least fixed point over the loop's characteristic function  $\Phi(X) = (1 - Pr(P_o)) \cdot F\, + \,Pr(P_o) \cdot \wp{C}{X}$. 
$\Phi$~is closely related to the weakest preexpectation of if-statements, as while loops are unrolled into if-statements during execution.
Calculating the least fixed point is undecidable \cite{rice1953classes}, which means that determining $\wp{C}{F}$ is also undecidable in general.
Still, the fixed point can be lower-bounded by repeatedly applying~$\Phi$ to the predicate~\mbox{$\mathbf{0} \coloneqq \lambda\, \beta\mydot 0$}. 
This can be seen as approximating the weakest preexpectation of a general while loop with the weakest preexpectation of bounded while loops with an ever increasing number of iterations. 
Formally, we even have:%
$$\wp{\WHILEDO{P_o}{C}}{F} = \lfp\Phi = \lim\limits_{n \rightarrow \infty} \Phi^n(\mathbf{0})$$

\paragraph{Composition and Skip.} 
Both statements follow the classical rules for $\wpsymbol$. 
The composition $\COMPOSE{C_1}{C_2}$ applies $\wpsymbol$ step-wise by first determining the weakest preexpectation for $C_2$ with the predicate $F$, and then uses the result as the predicate for the weakest preexpectation of $C_1$. 
The skip statement has no effect.

\subsection{Soundness}

As previously mentioned, our weakest preexpectation calculus shall capture the expected value of $F$ after executing the program $C$. 
This idea is formally captured in the following soundness theorem:%
\begin{theorem}[Soundness]\label{soundness}
	For all programs $C$, initial belief states $\beta_0$, and predicates $F$, we have
	$$ \wp{C}{F}\, (\beta_0) \qeq \sum_{\mathclap{\beta \in \mathcal{D}(\Sigma)}}~ [ C ]^{\beta_0}(\beta) \cdot F (\beta)~.$$
\end{theorem}%
\begin{proof}
	By structural induction over $C$. See \ifthenelse{\boolean{extended}}{\Cref{Soundness_proof}}{the extended version~\extendedVersion} for details. \qed
\end{proof}%
\noindent%
Note that the soundness theorem refers to the sub-distribution $[ C ]^{\beta_0}$. Thus, $\wpsymbol$ only considers terminal configurations, and we cannot gain significant insights into phenomenons which are not captured by terminal states.
However, divergence still affects the value of $\wp{C}{F} (\beta_0)$. 
The missing probability mass in~$[ C ]^{\beta_0}$ amounts to the probability that $C$ diverges, and is automatically weighted at $0$ in the soundness theorem, regardless of $F$.\\
Additionally, we will briefly discuss two useful properties which allows us to calculate $\wpsymbol$ more easily. 
We refer to \ifthenelse{\boolean{extended}}{\Cref{app:properties}}{the extended version~\extendedVersion} for a more formal introduction of these properties.
The first property is \textit{linearity}, which is characterized by the following equation for predicates $F, G$ and constant $\alpha \in \mathbb{R}_{{}\geq 0}$:
$$\wp{C}{\alpha\cdot F + G} \qeq \alpha \cdot \wp{C}{F} \pplus \wp{C}{G}$$
In particular, linearity enables \emph{decomposing} predicates, making $\wpsymbol$ computations amenable to parallelization.
The second property is \textit{independence}, which allows one to \enquote{skip} calculating $\wp{C}{F}$ for a predicate if the program $C$ does not affect the value of~$F$. 
For example, one can show that if the loop-free program $C$ does not modify any variables in the predicate $P$, then we have:
$$\wp{C}{Pr(P)} \eeq Pr(P)$$


\newbool{grammarFirst}
\setboolean{grammarFirst}{true}

\ifthenelse{\boolean{grammarFirst}}
{\subsection{Structuring Predicates}\label{sec:grammar}
	There is another hurdle to an efficient computation of $\wpsymbol$:
	If we apply the rules from \Cref{wp} naively, then the predicates we push backwards through the program would accumulate modifiers such as $\statesubst{t}{f}$ and we would end up with terms like $Pr(d = 1)|_{t=0}\statesubst{t}{f}\statesubst{d}{g}$, which quickly becomes infeasible for larger programs.
	To address this issue, we present the set of predicates $G$ as defined by the grammar in \Cref{GrammarFigure}, which can simplify terms like $F\statesubst{x}{f}$ efficiently.
	\begin{figure}[t]
		\begin{align*}
			E &\Coloneqq x \in\Varso \cup \Varsu \,|\, c \in \mathbb{N} \,|\, [P_u]  \,|\, E \oplus E & \oplus \in \{+, -, \cdot, /\}\\
			\Phi &\Coloneqq Ex(E) \sim q \cdot Ex(E) & \sim \in \{<,\leq, =, \neq, \geq, >\} \\
			G &\Coloneqq Ex(E) \,|\, G + G \,|\, \alpha \cdot G \,|\, [\Phi] \cdot G 
		\end{align*}
		$$Ex(E) \coloneqq \lambda\, \beta \mydot \sum_{\sigma \in \Sigma} \beta(\sigma)\cdot \sigma(E)$$
		\caption{An expressive grammar $G$ for predicates.}
		\label{GrammarFigure}
	\end{figure}
	Here, $E\colon \Sigma \rightarrow \mathbb{N}$ is a $G$-expression which additionally may include propositions $P$, which will be evaluated to either $0$ or $1$ in accordance with $\SMODEL{\sigma}{P}$. 
	The predicate $Ex(E)$ is defined as the expected value of $E$ within $\beta$. 
	Note that $Pr(P) = Ex([P])$, which will be useful when we consider $\wpsymbol$. 
	Furthermore, $G$ is expressive with respect to~$\wpsymbol$ for loop-free programs, as formalized below:
	\begin{theorem}[Expressiveness of $\boldsymbol{G}$]\label{GrammarNew}
		For any syntactically expressible loop-free program $C$ and predicate $F \in G$ we have $\wp{C}{F}\in G$
	\end{theorem}
	\begin{proof}
		By structural induction on $C$, see \ifthenelse{\boolean{extended}}
		{\Cref{GrammarProofNew}}
		{the extended version \extendedVersion}.\qed
	\end{proof}%

	\noindent
	There are two restrictions on in \Cref{GrammarNew}, namely the programs need to be \textit{loop-free} and \textit{syntactically expressible}.
	A grammar for programs with loops is beyond the scope of this paper, but loop-free programs still allow us to reason about bounded loops or loop invariants.
	Being syntactically expressible does not significantly restrict programs, but it still guarantees that programs are sufficiently structured to easily resolve predicate modifiers.
	We consider $C$ to be \textit{syntactically expressible}, if it adheres to the following restrictions:
	\begin{itemize}
		\item The conditions of infer statements may only feature such intervals $i$ that can be expressed using a comparison operator. 
		\item The functions for sample statements must be expressible as finite sums $f(\sigma) = \sum  p_i|\sigma(E_i)\rangle$ for some fixed $G$-expressions $E_i$ and probabilities $p_i$. 
	\end{itemize}

	\noindent%
	Notably, \Cref{GrammarNew} guarantees that any predicate modifier which is introduced during the computation of $\wpsymbol$ will be removed.
	To do so, every modifier applied to a predicate in $G$ will first be passed down recursively until we reach the base case $Ex(E)$.
	Then, any given modifier can be resolved syntactically as follows:

	\begin{lemma}[Resolution]
		Let $E, E'$ be $G$-expressions, $x$ be a variable, $c\in \mathbb{N}$ be a value and $f=\sum  p_i|E_i\rangle $ be the function of a sample statement.
		Then we have:
		\begin{align*}
			Ex(E)[x \mapsto E'] &\eeq Ex(E[x / E']) \tag{Assign}\label{lem:replacement}\\
			Ex(E)|_{x = c} &\eeq \frac{Ex(E\cdot[x=c])}{Ex([x = c])} \tag{Observe} \label{BeliefConditioningTheoremNew} \\
			Ex(E)\statesubst{x}{f} &\eeq Ex\left( E\leftarrow f\right) \tag{Sample} \label{BeliefElimTheoremNew} \\
		\end{align*}
		where $E[x/ E']$ is the $G$-expression $E$ with all occurrences of $x$ replaced by $E'$, and $E\leftarrow f \coloneqq \sum_{i \in \mathbb{N}} p_i \cdot E[x / E_i]$
	\end{lemma}
	\begin{proof}
		\Cref{lem:replacement} follows from definition of $Ex(E)$ and $\sigma[x\mapsto E'] (E) = \sigma (E[x / E'])$. 
		For \Cref{BeliefConditioningTheoremNew} and \Cref{BeliefElimTheoremNew} we refer to
		\ifthenelse{\boolean{extended}}
		{\Cref{ConditioningProofNew} and \Cref{BeliefElimProofNew}}
		{the extended version \extendedVersion}.
		Here, $E\leftarrow f$ builds upon Jacobs \cite{jacobs2019mathematics} to describe the outcome of a sampling as the expectation over all possible samples.
\end{proof}
}{}

\subsection{Reasoning about Loops with Loop Invariants for $\wpsymbol$}\label{CaseStudyLoopsSection}
We now showcase how $\wpsymbol$ can soundly reason about unbounded loops, a task which remains challenging for most PPLs.
The weakest preexpectation of a loop with respect to postexpectation $f$ is defined as the least fixed point of the loops characteristic function $\Phi_f$. 
This fixed point can be under-approximated by repeatedly applying $\Phi_f$ to the predicate $\mathbf{0}$, where each application yields a tighter \textit{lower} bound on the least fixed point, until we eventually converge (after possibly infinitely many iterations).
For \emph{upper} bounds, we can use a loop invariant:
Verifying that $I$ is an invariant for the loop $\WHILEDO{P_o}{C}$ is computationally cheaper than iteration, as we only apply $\Phi_f$ once and check following inequality:
$$ I \ggeq \Phi_f(I) \eeq (1-Pr(P_o)) \cdot f \pplus Pr(P_o) \cdot \wp{C}{I}$$
\begin{wrapfigure}[14]{l}{0.6\textwidth}
	\vspace*{-1 \intextsep}%
	\begin{lstlisting}
d = sample(0.9|1> + 0.1|0>);
inCare = true;
while(inCare){
	infer(p(d = 1) > q) {
		d = sample(0.75|d> + 0.25|0>);
		t = sample(0.95|d> + 0.05|1-d>);
		observe t;
	} else {
		inCare = false;
	}
}
\end{lstlisting}
	\caption{Probabilistic belief program $C_w$ 
	}
	\label{CaseStudyLoopyRepeat}
\end{wrapfigure}
Now, we reconsider the introductory example, restated in \Cref{CaseStudyLoopyRepeat} for convenience.
The program models a strategy for the ongoing treatment of a (possibly) sick patient, and the goal is to only release the patient once they are cured.
Hence, the probability to release a sick patient is highly relevant to us, and we can determine it through the weakest preexpectation of the predicate $Pr(d=1)$.
However, determining the value of $\wp{C_w}{Pr(d=1)}$ is difficult due to the loop, so instead we can show that the following predicate $I$ is an invariant of the loop:
$$ I \coloneqq (1 - Pr(\mathit{inCare})) \cdot Pr(d=1) \pplus Pr(\mathit{inCare}) \cdot q $$
This invariant encodes the two possible outcomes when entering the loop: Either $\mathit{inCare}$ is true and the patient takes treatment until the likelihood of the disease is at most $q$, or the patient never enters treatment and the likelihood of the disease is unchanged. 
\ifthenelse{\boolean{extended}}
{The proof that $I$ is an Invariant can be found in \Cref{wpCalc}.}
{We show $I$ to be an invariant in the extended version~\extendedVersion.}

\ifthenelse{\boolean{grammarFirst}}{}
{\subsection{Structuring Predicates}
There is another hurdle to an efficient computation of $\wpsymbol$:
If we apply the rules from \Cref{wp} naively, then the predicates we push backwards through the program would accumulate modifiers such as $\statesubst{t}{f}$ and we would end up with terms like $Pr(d = 1)|_{t=0}\statesubst{t}{f}\statesubst{d}{g}$, which quickly becomes infeasible for larger programs.
To address this issue, we present the set of predicates $G$ as defined by the grammar in \Cref{GrammarFigure}, which can simplify terms like $F\statesubst{x}{f}$ efficiently.
\begin{figure}[t]
	\begin{align*}
		E &\Coloneqq x \in\Varso \cup \Varsu \,|\, c \in \mathbb{N} \,|\, [P_u]  \,|\, E \oplus E & \oplus \in \{+, -, \cdot, /\}\\
		\Phi &\Coloneqq Ex(E) \sim q \cdot Ex(E) & \sim \in \{<,\leq, =, \neq, \geq, >\} \\
		G &\Coloneqq Ex(E) \,|\, G + G \,|\, \alpha \cdot G \,|\, [\Phi] \cdot G 
	\end{align*}
	$$Ex(E) \coloneqq \lambda\, \beta \mydot \sum_{\sigma \in \Sigma} \beta(\sigma)\cdot \sigma(E)$$
	\caption{An expressive grammar $G$ for predicates.}
	\label{GrammarFigure}
\end{figure}
Here, $E\colon \Sigma \rightarrow \mathbb{N}$ is a $G$-expression which additionally may include propositions $P$, which will be evaluated to either $0$ or $1$ in accordance with $\SMODEL{\sigma}{P}$. 
The predicate $Ex(E)$ is defined as the expected value of $E$ within $\beta$. 
Note that $Pr(P) = Ex([P])$, which will be useful when we consider $\wpsymbol$. 
Furthermore, $G$ is expressive with respect to~$\wpsymbol$ for loop-free programs, as formalized below:
\begin{theorem}[Expressiveness of $\boldsymbol{G}$]\label{GrammarNew}
	For any syntactically expressible loop-free program $C$ and predicate $F \in G$ we have $\wp{C}{F}\in G$
\end{theorem}
\begin{proof}
	By structural induction on $C$, see \Cref{GrammarProofNew}.\qed
\end{proof}%

\noindent
There are two restrictions on in \Cref{GrammarNew}, namely the programs need to be \textit{loop-free} and \textit{syntactically expressible}.
A grammar for programs with loops is beyond the scope of this paper, but loop-free programs still allow us to reason about bounded loops or loop invariants.
Being syntactically expressible does not significantly restrict programs, but it still guarantees that programs are sufficiently structured to easily resolve predicate modifiers.
We consider $C$ to be \textit{syntactically expressible}, if it adheres to the following restrictions:
\begin{itemize}
	\item The conditions of infer statements may only feature such intervals $i$ that can be expressed using a comparison operator. 
	\item The functions for sample statements must be expressible as finite sums $f(\sigma) = \sum  p_i|\sigma(E_i)\rangle$ for some fixed $G$-expressions $E_i$ and probabilities $p_i$. 
\end{itemize}

\noindent%
	Notably, \Cref{GrammarNew} guarantees that any predicate modifier which is introduced during the computation of $\wpsymbol$ will be removed.
	To do so, every modifier applied to a predicate in $G$ will first be passed down recursively until we reach the base case $Ex(E)$.
	Then, any given modifier can be resolved syntactically as follows:

\begin{lemma}[Resolution]
	Let $E, E'$ be $G$-expressions, $x$ be a variable, $c\in \mathbb{N}$ be a value and $f=\sum  p_i|E_i\rangle $ be the function of a sample statement.
	Then we have:
	\begin{align*}
		Ex(E)[x \mapsto E'] &\eeq Ex(E[x / E']) \tag{Assign}\label{lem:replacement}\\
		Ex(E)|_{x = c} &\eeq \frac{Ex(E\cdot[x=c])}{Ex([x = c])} \tag{Observe} \label{BeliefConditioningTheoremNew} \\
		Ex(E)\statesubst{x}{f} &\eeq Ex\left( E\leftarrow f\right) \tag{Sample} \label{BeliefElimTheoremNew} \\
\end{align*}
where $E[x/ E']$ is the $G$-expression $E$ with all occurrences of $x$ replaced by $E'$, and $E\leftarrow f \coloneqq \sum_{i \in \mathbb{N}} p_i \cdot E[x / E_i]$
\end{lemma}
\begin{proof}
\Cref{lem:replacement} follows from definition of $Ex(E)$ and $\sigma[x\mapsto E'] (E) = \sigma (E[x / E'])$. 
For \Cref{BeliefConditioningTheoremNew} and \Cref{BeliefElimTheoremNew} we refer to \Cref{ConditioningProofNew} and \Cref{BeliefElimProofNew}. 
Here, $E\leftarrow f$ builds upon Jacobs \cite{jacobs2019mathematics} to describe the outcome of a sampling as the expectation over all possible samples.
\end{proof}}


\section{Related Work}

A very common approach to modeling partial observability is to use a POMDP \cite{smallwood1973optimal,durbin1998biological,kress2009temporal}.
However, POMDPs are based on an explicit representation of the state space, which is in stark contrast to the symbolic approach of \PBLIMP and other PPLs.
In particular, the symbolic approach of PPLs allows one to compactly model and reason about state spaces without having to consider each state individually.
Furthermore, Evans et al. \cite{agentmodels} demonstrated how probabilistic programs comparable to \PBLIMP can model POMDPs.
Using a similar approach, \PBLIMP can be used as a language to encode certain POMDPs.
We will not compare POMDPs and PPLs as a means to model partial observability in detail, but we refer to Evans et al. \cite{agentmodels} for an in-depth description.
Instead, we would like to focus more thoroughly on the relation between \PBLIMP and other PPLs.

$\wpsymbol$-calculi in the style of \cite{DBLP:series/mcs/McIverM05} have long been used as a tool for formal verification of program correctness, and our work builds upon this idea.
In the context of PPLs, our $\wpsymbol$ can be considered an inference method 
for probabilistic programs with nested inference, i.e.\ inference about programs which themselves perform inferences.
However, not all PPLs feature nested inference, and some languages \cite{holtzen2020scaling,randone2024inference} trade off the lower expressiveness for faster inference.
Still, PPLs with nested inference can often model the same systems as \PBLIMP, even though their observe statement behaves differently.
Additionally, these languages commonly support continuous distributions, and they may feature richer interaction with the inferred distribution than \PBLIMP's infer statement \cite{baudart2020reactive,tran2017deep,goodman2012church}.

There is a variety of PPLs which use sampling algorithms to perform \textit{approximate} nested inference \cite{dippl,tolpin2016design,tavares2019random,baudart2020reactive}. 
Sampling is often faster than exact inference, but it is difficult to get sound bounds on the quality of the approximation \cite{cowles1996markov,chatterjee2018sample}.

In contrast to sampling, the $\wpsymbol$ calculus performs exact symbolic inference, which can also be performed with Bhat et al \cite{bhat2013deriving} or tools like Hakaru \cite{narayanan2016probabilistic}, and some PPLs like WebPPL \cite{dippl} support inference for discrete distributions by enumerating all paths.
But, to our knowledge, the only exact symbolic inference tool with nested inference is PSI \cite{gehr2020lpsi}. 
PSI is able to analyze programs containing both discrete and continuous distributions, while providing simple descriptions of the inferred distribution.
PSI (and most PPLs) does not support unbounded loops, thus programs like \Cref{CaseStudyLoopy} need to be approximated with bounded loops.
This yields either an under-approximation or no guarantees at all, and bounded loops provide no guarantees regarding the accuracy of the final result \cite{torres2024iteration}.

Therefore, loop invariants as seen with \Cref{CaseStudyLoopyRepeat} are a key distinction between $\wpsymbol$ and existing inference methods for PPLs.
Loop invariants are well-established in the field of program analysis \cite{kaminski2019advanced}, but to our knowledge they have not been applied to probabilistic programs with both conditioning and nested inference.
Loop invariants yield a sound over-approximation of the result, but they usually need to be provided by the developer.
Still, loop invariants for probabilistic programs without conditioning have successfully been automatically synthesized~\cite{batz2023probabilistic}.

\section{Conclusion and Future Work}
Our work  thoroughly laid the theoretical groundwork for probabilistic belief programming, a methodology to write programs for partially observable environments. 
For this, we introduced the imperative probabilistic belief programming language \PBLIMP and provided operational semantics for it. 
We introduced a weakest preexpectation calculus -- sound with respect to the operational semantics -- which can be used to reason about quantitative properties of \PBLIMP programs.
Furthermore, we demonstrated how weakest preexpectations can be calculated in an effective manner, including a treatment of loops.
Nonetheless, there are still challenges for a practical implementation of our work.
In particular, the automatic derivation of loop invariants \cite{batz2023probabilistic,DBLP:journals/pacmpl/SchroerBKKM23} is a crucial feature for the analysis of loops.
In future work we would like to address this gap and implement the efficient and fully automatic analysis of \PBLIMP programs.

	%
	%
	%
	 \bibliographystyle{splncs04}
	\bibliography{bibliographyPBLIMP}
	
	\newpage
	
	\appendix
	\section{Alternative Operational Semantics}\label{semantics_full}
Below we present full operational semantics without a true variable assignment, allowing for easier reasoning in the proof of \Cref{soundness}.
A program configuration is alternatively a triple $\langle C, \beta, p \rangle$, where $C$, $\beta$, and $p$ are interpreted as before, and successor configurations are defined with a small-step operational approach.
The largest difference is that the true variable assignment $\sigma$ is missing in such alternative program configurations.
The semantics are hence very similar to the rules from \Cref{Transitions}, with the exception of the observe and sample statements.

\renewcommand{\arraystretch}{3.4}
\begin{adjustbox}{max width=\textwidth}
	\begin{tabular}{c c}
		\inference{ }{  \langle \SKIP,\, \beta, \, p \rangle  \qtriangleright  \langle \top,\, \beta,\, p \rangle } &
		
		\inference{\beta\statesubst{x_o}{E_o}(\sigma') \coloneqq\sum\limits_{
				\rho\statesubst{x_o}{E_o} \eeq \sigma'} \beta(\rho)}
		{ \langle \ASSIGN{x_o}{E_o},\, \beta,\, p \rangle \qtriangleright \langle \top,\,\beta\statesubst{x_o}{E_o},\, p \rangle }\\
		\inference{}
		{ \langle \CHOOSE{x_u}{f},\,\beta,\, p \rangle \qtriangleright \langle \top,\,\beta\statesubst{x_u}{f},\, p \rangle } &
		
		\inference{\BMODEL{\beta}{ x_u \eeq c} \eeq q \qquad q > 0}
		{  \langle \OBSERVEASSIGN{y_o}{x_u},\, \beta,\, p \rangle \qtriangleright \langle \top,\, \beta|_{x_u=c}\statesubst{y_o}{c},\, p\cdot q \rangle } \\
		
		\inference{(\BMODEL{\beta}{P_o}) \eeq 1}
		{ \langle \ITE{P_o}{C_1}{C_2}\, \beta,\, p \rangle \qtriangleright \langle C_1,\,\beta,\, p \rangle } &
		
		\inference{(\BMODEL{\beta}{P_o}) \eeq 0}
		{ \langle \ITE{P_o}{C_1}{C_2},\, \beta,\, p \rangle \qtriangleright \langle C_2,\,\beta,\, p \rangle } \\
		\multicolumn{2}{c}{\inference{}
			{ \langle \WHILEDO{P_o}{C},\,\beta,\, p \rangle \qtriangleright \langle \ITE{P_o}{\COMPOSE{C}{ \WHILEDO{P_o}{C}}}{\SKIP},\, \beta,\, p \rangle }}\\
		
		\inference{(\BMODEL{\beta}{P_u}) \in i}
		{ \langle \INFER{P_u \in i}{C_1}{C_2},\, \beta,\, p \rangle \qtriangleright \langle C_1,\,\beta,\, p \rangle } &
		
		\inference{\langle C_1,\, \beta,\, 1 \rangle \qtriangleright \langle C_1',\, \beta',\, q \rangle \qquad C_1' \neq \top
		} 
		{ \langle \COMPOSE{C_1}{C_2},\, \beta,\, p \rangle \qtriangleright \langle \COMPOSE{C_1'}{C_2},\, \beta',\, p \cdot q \rangle }\\
		
		\inference{(\BMODEL{\beta}{P_u}) \notin i}
		{ \langle \INFER{P_u \in i}{C_1}{C_2},\, \beta,\, p \rangle \qtriangleright \langle C_2,\,\beta,\, p \rangle } &
		
		\inference{\langle C_1,\, \beta,\, 1 \rangle \qtriangleright \langle \top,\, \beta',\, q \rangle}
		{ \langle \COMPOSE{C_1}{C_2},\, \beta,\, p \rangle \qtriangleright \langle C_2,\, \beta,\, p \cdot q\rangle }\\
	\end{tabular}
\end{adjustbox}

					\paragraph{Observe and Sample.} 
					There is a notable difference between the two operational semantics in the interaction between the observe and sample statement:
					Under the first semantics, a configuration branches out into multiple possible successor configurations when we execute a sample statement, whereas the observe statement only allows for one possible successor configuration. 
					There, \emph{the sample statement} \enquote{picks} the true value immediately and updates the true variable assignment accordingly. 
					In the second -- alternative~-- operational semantics, the sample statement only has one possible successor configuration, whereas the observe statement branches out into multiple successors. 
					Here, \emph{the observe statement} \enquote{picks} the true value, and we need a successor configuration for all the possible observations. 
					In both cases, the same updates are made to the respective belief states with the same probabilities when we transition between configurations.

					\paragraph{Equivalence between the Two Operational Semantics.}
					We already showed how the computation tree rooted in a configuration $(C, \beta, \sigma, 1)$ can be used to generate a sub-distribution $[C]^{\beta, \sigma}$ over the terminal leaves. 
					In a fully analogous manner, we can use the computation tree rooted in the configuration $\langle C,\beta, 1 \rangle$ to generate a sub-distribution $\llbracket C\rrbracket^{\beta}$ over the terminal belief states.
					We will use $\llbracket C\rrbracket^{\beta}(\beta')$ to denote the probability that executing $C$ from a belief state $\beta$ terminates in $\beta'$.
					
					At this point, one might ask how the two semantics relate. 
					As it turns out, both approaches agree on the probability to reach a particular belief state, as is formalized in the following theorem:%
					\begin{theorem}[Operational Correctness]\label{OperationalCorrectnessTheorem}
							Let $C$ a program, $\beta_0$ an initial belief state, $\beta$ a final belief state, and $\sigma$ a variable assignment.
							Then we have:
							\begin{enumerate}
									\item $[C]^{\beta_0}(\beta) \eeq \llbracket C\rrbracket^{\beta_0} (\beta)$
									\item $[C]^{\beta_0} (\beta, \sigma) \eeq \beta(\sigma) \cdot \llbracket C \rrbracket^{\beta_0} (\beta)$
								\end{enumerate}
						\end{theorem}%
					\begin{proof}%
							(1)\ follows by structural induction over $C$, see \Cref{Tree_proof}. 
							(2)\ follows directly from \Cref{beliefAccuracyTheorem} and (1).
							\qed
						\end{proof}%

\section{Proof of Preservation of Belief State Accuracy}\label{belief_accuracy}
Here, we formally prove \Cref{beliefAccuracyTheorem}. For convenience, the theorem is repeated below:

Let $C$ be a program, and $\beta_0$ the initial belief state and $\sigma$ a variable assignment. Then we have:
$$[C]^{\beta_0} (\beta, \sigma) \eeq \beta(\sigma) \cdot [C]^{\beta_0} (\beta)$$
\begin{proof}
	We prove the statement by induction on $C$, giving us the following cases:
	\subsection{\SKIP}
	Every initial configuration $(\SKIP, \beta_0, \sigma_0, 1)$ reduces to $(\top, \beta_0, \sigma_0, 1)$. Thus, we have $\beta_0$ as the only reachable belief state with $[\SKIP]^{\beta_0} (\beta_0) = 1$, and the theorem reduces to $[\SKIP]^{\beta_0} (\beta_0, \sigma) = \beta_0(\sigma)$. This follows directly from the transition rules and the definition of $[\SKIP]^{\beta_0}$.
	\subsection{\ASSIGN{x}{E}}
	Similarly to the skip statement, we only have a single reachable belief state $\beta_0\statesubst{x}{E}$ with $[\ASSIGN{x}{E}]^{\beta_0} (\beta_0\statesubst{x}{E}) = 1$, and the theorem again reduces to $[\ASSIGN{x}{E}]^{\beta_0} (\beta_0\statesubst{x}{E}, \sigma) = \beta_0\statesubst{x}{E}(\sigma)$ which follows directly from the transition rules and the definition of $\beta_0\statesubst{x}{E}$, as every initial configuration generates exactly one terminal state.
	\subsection{\ttsymbol{if}  and \ttsymbol{infer}}
	For both the if-statement as well as the infer statement, all relevant initial configurations will take the same branch independent of the particular real variable assignment. This holds, as both the infer statement and the if-statement pick their branch depending only on the belief state $\beta_0$.
	
	Thus, we can assume w.l.o.g. that $\BMODEL{\beta_0}{P_o} = 1$ holds, and we can conclude that we have: 
	$$[\ITE{P_o}{C_1}{C_2}]^{\beta_0} = [C_1]^{\beta_0}$$
	where $C_i$ is the Code of the branch we picked to execute. The theorem now follows immediately from the inductive hypothesis applied to $C_1$. The proof for the case $\BMODEL{\beta_0}{P_o} = 0$ and the infer statement are fully analogous.
	\subsection{\ttsymbol{while}}
	We will first consider the theorem for bounded loops, defined as
	$$\BOUNDEDWHILE{n}{P_o}{C} \coloneqq \begin{cases} \diverge &, n = 0\\
		\ITE{P_o}{\COMPOSE{C}{\BOUNDEDWHILE{n-1}{P_o}{C}}}{\SKIP} &, n>0
	\end{cases}$$
	Where the special statement diverge never terminates. The idea is to derive the behavior of unbounded loops as the limit of the behavior of bounded loops with ever increasing iterations. Formally expressed, this means:
	$$\lim\limits_{n \rightarrow \infty} [\BOUNDEDWHILE{n}{P}{C}]^{\beta_0} (\beta, \sigma) =  [\WHILEDO{P}{C}]^{\beta_0} (\beta, \sigma) $$
	and 
	$$\lim\limits_{n \rightarrow \infty} [\BOUNDEDWHILE{n}{P}{C}]^{\beta_0} (\beta) = [\WHILEDO{P}{C}]^{\beta_0} (\beta) $$
	We only prove the first equation as the second one follows from the first equation. To be precise, we will first show that for the execution of the first $n$ loops iterations, the terminated leaves on the computation trees for $(\BOUNDEDWHILE{n}{P}{C}, \beta, \sigma, 1)$ and $(\WHILEDO{P}{C}, \beta, \sigma, 1)$ are identical. For our purposes we consider an iteration of the loop to start with checking the condition, and to end when the loop body or skip has been executed. 
	\begin{proof}Proof by induction:
		\begin{itemize}
			\item [$n=0$:] The tree for $\BOUNDEDWHILE{0}{P}{C}$ has no terminated leaves, as by definition we have $\BOUNDEDWHILE{0}{P}{C} \eeq \diverge$. Additionally, \WHILEDO{P}{C} also has no terminated leaves by our our definition of a loop iteration. The theorem holds trivially.
			\item [$n>0$:] We are now able to unroll both while loops into an if-statement. Thus we get the following statements:\\
			\renewcommand{\arraystretch}{2.5}
			\begin{tabular}{r c l}
				$\BOUNDEDWHILE{n}{P}{C}$ & turns to & $\ITE{P_o}{\COMPOSE{C}{\BOUNDEDWHILE{n-1}{P_o}{C}}}{\SKIP}$\\
				$\WHILEDO{P}{C}$ & turns to & $\ITE{P_o}{\COMPOSE{C}{\WHILEDO{P}{C}}}{\SKIP}$
			\end{tabular}\\ 
			We make a case distinction on whether $\BMODEL{\beta_0}{P} = 0$ or $\BMODEL{\beta_0}{P} = 1$ holds, as one has to hold, and resolve the if-statement:
			\begin{itemize}
				\item $\BMODEL{\beta_0}{P} = 0$: In this case both $\BOUNDEDWHILE{n}{P}{C}$ and $\WHILEDO{P}{C}$ behave like $\SKIP$ and create the same terminal leaf after one transition
				\item $\BMODEL{\beta_0}{P} = 1$: In this case both the bounded and unbounded loop take the same branch of the if statement, 
				and it remains to show that $\COMPOSE{C}{\BOUNDEDWHILE{n-1}{P_o}{C}}$ and $\COMPOSE{C}{\WHILEDO{P}{C}}$ create the same terminated leafs for the first $n$ loop iterations.
				
				Both programs first execute $C$ and terminate in identical intermediate belief states and variable assignments with the same probabilities. This is the first loop iteration. Afterwards, we execute $\BOUNDEDWHILE{n-1}{P_o}{C}$ or $\WHILEDO{P}{C}$ starting from those intermediate belief states. However, by the inductive hypothesis, this will create identical leaves for the next $n-1$ iterations, thus leading to the same terminated leaves for in total $n$ iterations.
			\end{itemize} 
		\end{itemize}
	\end{proof}
	Now assume $\lim_{n \rightarrow \infty} [\BOUNDEDWHILE{n}{P}{C}]^{\beta_0} (\beta, \sigma) \neq [\WHILEDO{P}{C}]^{\beta_0} (\beta, \sigma)$. Thus, there has to be a terminal leaf within either $\lim_{n \rightarrow \infty} [\BOUNDEDWHILE{n}{P}{C}]^{\beta_0}$ or $[\WHILEDO{P}{C}]^{\beta_0}$ which does not have a counterpart in the other one. However, that leaf has to be reached after a finite number of loop iterations, and thus it also has to be present in the other computation tree. We have reached a contradiction, and thus the assumption cannot be true.
	
	It now suffices to prove $[\BOUNDEDWHILE{n}{P}{C}]^{\beta_0} (\beta, \sigma) = \beta(\sigma) \cdot [\BOUNDEDWHILE{n}{P}{C}]^{\beta_0} (\beta)$, 
	proof by induction over $n$:
	\begin{itemize}
		\item [$n = 0$:] This case is trivial as there are no terminated configurations and the result is thus $0$ on both sides
		\item [$n > 0$:] We have $\BOUNDEDWHILE{n}{P}{C} = \ITE{P_o}{\COMPOSE{C}{\BOUNDEDWHILE{n-1}{P_o}{C}}}{\SKIP}$. If $\BMODEL{\beta_0}{P} = 0$ the expression reduces to $\SKIP$, which can again be proven as was done in the case for $\SKIP$. Otherwise we have $\BMODEL{\beta_0}{P} = 1$, and
		the theorem follows with the same argumentation as was done in the case for $\COMPOSE{C_1}{C_2}$ and the inductive hypothesis.
	\end{itemize}
	\subsection{$\COMPOSE{C_1}{C_2}$}
	Note that we can divide the execution of $\COMPOSE{C_1}{C_2}$ into a prefix which executes $C_1$ and terminates in some intermediate belief state $\beta'$, and a suffix which executes $C_2$ starting from that particular $\beta'$. However, hidden in the definition of $[C_2]^{\beta'}$ is the weighted sum over all possible initial variable assignments, so that we can then consider all possible computation trees.
	
	Thus, the weighted sum over initial computation trees $[C_2]^{\beta', \sigma'}$ for $C_2$ has to match with the probabilities to reach the terminal configurations $(\top, \beta', \sigma', p)$ when executing $C_1$. This amounts to proving $[C_1]^{\beta_0} (\beta', \sigma') = \beta'(\sigma') \cdot [C_1]^{\beta_0} (\beta')$, which holds in accordance with the inductive hypothesis for $C_1$.
	Thus, we can express $[\COMPOSE{C_1}{C_2}]^{\beta_0}$ using $[C_1]^{\beta_0}$ and $[C_2]^{\beta'}$ by considering all possible intermediate belief states $\beta'$. We now get:  
	\begin{align*}
		[\COMPOSE{C_1}{C_2}]^{\beta_0} (\beta, \sigma) &= \sum_{\beta' \in \mathcal{D}(\Sigma)} [C_1]^{\beta_0} (\beta') \cdot [C_2]^{\beta'} (\beta, \sigma) \\
		&=_{IH} \sum_{\beta' \in \mathcal{D}(\Sigma)} [C_1]^{\beta_0} (\beta') \cdot (\beta(\sigma) \cdot [C_2]^{\beta'} (\beta)) \\
		&= \beta(\sigma) \cdot \sum_{\beta' \in \mathcal{D}(\Sigma)} [C_1]^{\beta_0} (\beta') \cdot [C_2]^{\beta'} (\beta) \\
		&= \beta(\sigma) \cdot [\COMPOSE{C_1}{C_2}]^{\beta_0} (\beta)
	\end{align*}

	\subsection{\CHOOSE{x}{f}}
	For $[\CHOOSE{x}{f}]^{\beta_0}$ the only reachable terminal belief state is $\beta_0\statesubst{x}{f}$, we thus have $[\CHOOSE{x}{f}]^{\beta_0}(\beta\statesubst{x}{f}) = 1$. Furthermore the probability to reach a leaf with the real variable assignment $\sigma$ is 
	$$\sum\limits_{\rho \in \Sigma, \rho[x\mapsto \sigma(x)] = \sigma} f(\rho, \sigma(x))\cdot \beta_0(\rho)$$ 
	This can be understood as the probability to start in any $\rho$ which could turn into $\sigma$ after updating $x$ to the correct value, and then the probability for $f$ to "sample" the right value, i.e.\ $f(\rho,\sigma(x))$. That expression by definition is equal to $\beta_0\statesubst{x}{f} (\sigma)$. We thus get:
	\begin{align*}
		&[\CHOOSE{x}{f}]^{\beta_0} (\beta_0\statesubst{x}{f}, \sigma) \\
		&= \sum_{\rho \in \Sigma, \rho[x\mapsto \sigma(x)] \eeq \sigma} f(\rho,\sigma(x))\cdot \beta_0(\rho)\\
		&= \beta_0\statesubst{x}{f}(\sigma)\\
		&= \beta_0\statesubst{x}{f}(\sigma) \cdot 1\\
		&= \beta_0\statesubst{x}{f}(\sigma) \cdot [\CHOOSE{x}{f}]^{\beta_0}(\beta_0\statesubst{x}{f})
	\end{align*}

	\subsection{\OBSERVEASSIGN{y}{x}}
	The observe statement can be split into two parts: A "pure" observe statement \OBSERVE{x} which only conditions the belief state via $\beta_0|_{x=c}$, and an assignment which updates the value of $y$ with the (now deterministic) value of $x$.
	It suffices to consider the first part, i.e. $\OBSERVE{x}$, as the second part can be proven fully analogous to the case for the assignment, because after the observe statement has been executed there is only a single possible value for $x$ in the belief state.
	Afterwards the full observe statement follows from the proofs for $\OBSERVE{x}$, the assignment and sequential execution.
	
	All reachable belief states have the form $\beta_0|_{x=c}\statesubst{y}{c}$. 
	By definition of the transition relation, $\beta_0|_{x=c}$ is reached when $\sigma(x) = c$ holds in the real variable assignment. The probability to terminate with the belief state $\beta|_{x=c}$ can thus be derived from the probability to start with a variable assignment where $\sigma(x) = c$ holds. By the definition of $[C]^{\beta}$ we get: 
	$$[\OBSERVE{x}]^{\beta_0} (\beta_0|_{x=c}) =\sum_{\sigma\in \Sigma} \beta_0(\sigma)\cdot \SMODEL{\sigma}{x=c} \eeq \BMODEL{\beta_0}{x = c}$$ 
	We distinguish two cases: Belief states $\beta_0|_{x=c}$ with $\BMODEL{\beta_0}{x = c} = 0$ and belief states $\beta_0|_{x=c}$ with $\BMODEL{\beta_0}{x = c} > 0$. For the first case the theorem follows trivially, as we have $[\OBSERVE{x}]^{\beta_0} (\beta_0|_{x=c}) =  0$, because no possible initial variable assignment satisfies $\sigma(x) = c$ and thus $\beta_0|_{x=c}$ is unreachable. In the other case, we additionally assume that $\sigma(x) = c$ holds, as otherwise the theorem holds trivially because we have 
	$$[\OBSERVE{x}]^{\beta_0} (\beta_0|_{x=c}, \sigma) =0  =  \beta_0|_{x=c}(\sigma) \cdot [\OBSERVE{x}]^{\beta_0} (\beta_0|_{x=c})$$ 
	by definition of $\beta|_{x=c}$, the transition rule for observe and $\beta_0|_{x=c} \neq \beta_0|_{x=\sigma(x)}$ . Finally, we get:
	\begin{align*}
		[\OBSERVE{x}]^{\beta_0} (\beta_0|_{x=c}, \sigma) &= \sum\limits_{\sigma_0 \in \Sigma} \beta_0(\sigma_0) \cdot [\OBSERVE{x}]^{\beta_0, \sigma_0}(\beta_0|_{x=c}, \sigma)\\
		&=_1 \beta_0(\sigma) \cdot [C]^{\beta_0, \sigma}(\beta_0|_{x=c}, \sigma)\\
		&=_2 \beta_0(\sigma) \\
		&=\frac{\beta_0(\sigma)}{\BMODEL{\beta_0}{x = c}} \cdot \BMODEL{\beta_0}{x = c} \\
		&=\frac{\beta_0(\sigma) \cdot \SMODEL{\sigma}{x=c}}{\BMODEL{\beta_0}{x = c}} \cdot \BMODEL{\beta_0}{x = c} \\
		&=\beta_0|_{x=c}(\sigma) \cdot \BMODEL{\beta_0}{x = c} \\
		&=\beta_0|_{x=c}(\sigma) \cdot [\OBSERVE{x}]^{\beta_0} (\beta_0|_{x=c})
	\end{align*}
	The equation $=_1$ is sound because the observe statement does not modify the true variable assignment, and thus the term $[\OBSERVE{x}]^{\beta_0, \sigma_0}(\beta_0|_{x=c}, \sigma)$ is $0$ unless we have $\sigma_0 = \sigma$. The equation $=_2$ follows from $[\OBSERVE{x}]^{\beta_0, \sigma}(\beta_0|_{x=c}, \sigma) = 1$, which is a consequence of the transition rule for observe and our assumption $\sigma(x) = c$.
\end{proof}

\section{Proof of Operational Correctness}\label{Tree_proof}
Here, we formally prove \Cref{OperationalCorrectnessTheorem}. For convenience, the theorem is repeated below:

Let $C$ be a program, $\beta_0$ the initial belief state, and $\beta$ another belief state. Then we have:
$$[C]^{\beta_0}(\beta) \eeq \llbracket C\rrbracket^{\beta_0} (\beta)$$
\begin{proof}
	Proof by structural induction over $C$, giving us the following cases:
	\subsection{\SKIP}
	If we consider both computation trees, we realize that only the belief state $\beta_0$ is reachable in the terminal configurations, we thus have $[\SKIP]^{\beta_0}(\beta_0) = 1$ and $\llbracket \SKIP\rrbracket^{\beta_0} (\beta_0) = 1$. The theorem now follows directly.
	\subsection{\ASSIGN{x}{E}}
	We only have one reachable belief state, namely the belief state $\beta_0\statesubst{x}{E}$ with $[\ASSIGN{x}{E}]^{\beta_0}(\beta_0\statesubst{x}{E}) = 1$ and $\llbracket \ASSIGN{x}{E}\rrbracket^{\beta_0} (\beta_0\statesubst{x}{E}) = 1$. The theorem now follows directly.
	\subsection{\ITE{P_o}{C_1}{C_2}}
	As $P_o$ may only use observable variables, and as we only consider consistent belief states, we must also have $\BMODEL{\beta}{P_o} = 1$ or $\BMODEL{\beta}{P_o} = 0$, we assume w.l.o.g. $\BMODEL{\beta}{P_o} = 1$.
	
	As the if-statement only checks the condition with the belief state, all initial configurations for $[\IFSYMBOL\dots]^{\beta_0}$ take the same branch of the if-statement, and make a transition from $(\IFSYMBOL\dots, \beta_0, \sigma, 1)$ to $(C_1, \beta, \sigma, 1)$. Note that after this transition we arrive at the initial configurations for the forest of $[C_1]^{\beta_0}$, it should be clear that $[\IFSYMBOL\dots]^{\beta_0} (\beta)= [C_1]^{\beta_0} (\beta)$ holds.
	
	For $\llbracket \IFSYMBOL\dots\rrbracket^{\beta_0}$, we then have $\langle \IFSYMBOL\dots, \beta, 1 \rangle \triangleright \langle C_1,\beta, 1 \rangle$ as the only successor, thus the trees for $C$ and $C_1$ have the same leafs and create the same partial distributions. The theorem now follows from the inductive hypothesis for $C_1$.
	\subsection{while ($P_o$) C}
	We will first consider the theorem for bounded loops, defined as 
	$$\BOUNDEDWHILE{n}{P_o}{C} \coloneqq \begin{cases} \diverge &, n = 0\\
		\ITE{P_o}{\COMPOSE{C}{\BOUNDEDWHILE{n-1}{P_o}{C}}}{\SKIP} &, n>0
	\end{cases}$$
	This is possible as we have 
	$$\lim\limits_{n \rightarrow \infty} \llbracket \BOUNDEDWHILE{n}{P_o}{C}\rrbracket^{\beta_0} (\beta) = \llbracket \WHILEDO{P}{C}\rrbracket^{\beta_0} (\beta) $$
	and 
	$$\lim\limits_{n \rightarrow \infty} [\BOUNDEDWHILE{n}{P_o}{C}]^{\beta_0}\, (\beta) = [\WHILEDO{P}{C}]^{\beta_0}\, (\beta) $$
	The second equation has already been proven in \Cref{belief_accuracy}. We will prove the first equation by first showing that for the execution of the first $n$ loops iterations, the terminated leaves on the trees for $\llbracket \BOUNDEDWHILE{n}{P_o}{C}\rrbracket^{\beta_0}$ and $\llbracket \WHILEDO{P}{C}\rrbracket^{\beta_0}$ are identical. We consider a loop iteration to start when the loop condition is checked, and to end just before the condition would be checked again. We will do a proof by induction on $n$:
	\begin{itemize}
		\item [$n=0$:]  $\BOUNDEDWHILE{0}{P_o}{C}$ has no terminated leaves, as by definition of bounded loops we have $\BOUNDEDWHILE{0}{P_o}{C} = \diverge$. 
		Similarly the tree for $\WHILEDO{P}{C}$ also has no leaves which terminate within $0$ loop iterations by our definition of a loop iteration. The theorem holds trivially.
		\item [$n > 0$:] $\BOUNDEDWHILE{n}{P_o}{C}$ reduces to $\ITE{P_o}{\COMPOSE{C}{\BOUNDEDWHILE{n-1}{P_o}{C}}}{\SKIP}$ whereas the unbounded loop \WHILEDO{P}{C} behaves like the statement $\ITE{P_o}{\COMPOSE{C}{\WHILEDO{P}{C}}}{\SKIP}$. 
		We make a case distinction on whether $\BMODEL{\beta_0}{P} = 0$ or $\BMODEL{\beta_0}{P} = 1$ holds, as one has to hold, and resolve the if-statement:
		\begin{itemize}
			\item $\BMODEL{\beta_0}{P} = 0$: In this case both $\BOUNDEDWHILE{n}{P_o}{C}$ and $\WHILEDO{P}{C}$ behave like $\SKIP$ and create the same leaf after one transition.
			\item $\BMODEL{\beta_0}{P} = 1$: In this case the computation tree for $\BOUNDEDWHILE{n}{P_o}{C}$ behaves like $\COMPOSE{C}{\BOUNDEDWHILE{n-1}{P_o}{C}}$ and $\WHILEDO{P}{C}$ behaves like $\COMPOSE{C}{\WHILEDO{P_o}{C}}$. It remains to show that \COMPOSE{C}{\BOUNDEDWHILE{n-1}{P_o}{C}} and $\COMPOSE{C}{\WHILEDO{P_o}{C}}$ create the same terminated leafs for the first $n$ loop iterations. Both programs first execute $C$ and terminate in identical intermediate belief states with the same probabilities. This is the first loop iteration. Afterwards, we execute $\BOUNDEDWHILE{n-1}{P_o}{C}$ or $\WHILEDO{P}{C}$ starting from those intermediate belief states. However, by the inductive hypothesis, this will create identical leaves for the next $n-1$ iterations, thus leading to the same terminated leaves for a total of $n$ iterations.
		\end{itemize} 
	\end{itemize}
	Now assume $\lim_{n \rightarrow \infty} \llbracket \BOUNDEDWHILE{n}{P_o}{C}\rrbracket^{\beta_0} (\beta) \neq \llbracket \WHILEDO{P}{C}\rrbracket^{\beta_0} (\beta)$. Thus, there has to be some terminated leaf in either $\lim_{n \rightarrow \infty} \llbracket \BOUNDEDWHILE{n}{P_o}{C}\rrbracket^{\beta_0}$ or $\llbracket \WHILEDO{P}{C}\rrbracket^{\beta_0}$ which does not exist in the other one. However, that particular leaf has to be reached after a finite number of loop iterations, and thus we just showed that the leaf also occurs in the other computation tree. We conclude that the assumption cannot hold.
	In order to prove the original theorem, it suffices to show $\llbracket \BOUNDEDWHILE{n}{P_o}{C}\rrbracket^{\beta_0} (\beta) = [\BOUNDEDWHILE{n}{P_o}{C}]^{\beta_0} (\beta)$. We again conduct a proof by induction over $n$:\\
	\begin{itemize}
		\item [$n=0$:] Trivial as both diverge and thus both evaluate to $0$
		\item [$n>0$:] We have $\BOUNDEDWHILE{n}{P_o}{C} = \ITE{P_o}{\COMPOSE{C}{\BOUNDEDWHILE{n-1}{P_o}{C}}}{\SKIP}$. If $\BMODEL{\beta_0}{P} = 0$ the expression reduces to $\SKIP$, which can again be proven as was done in the case for $\SKIP$. Otherwise we have $\BMODEL{\beta_0}{P} = 1$, and we need to show: 
		$$\llbracket \COMPOSE{C}{\BOUNDEDWHILE{n-1}{P_o}{C}}\rrbracket^{\beta_0} (\beta) = [\COMPOSE{C}{\BOUNDEDWHILE{n-1}{P_o}{C}}]^{\beta_0} (\beta)$$
		The equation above follows directly by using the same argumentation as in the case for $\COMPOSE{C_1}{C_2}$ and the inductive hypothesis using $n-1$.
	\end{itemize}
	\subsection{$\COMPOSE{C_1}{C_2}$}
	The execution of $\COMPOSE{C_1}{C_2}$ can be divided into a prefix which executes $C_1$ and terminates in a belief state $\beta'$, and then continuing with the execution of $C_2$ from there. We thus have:
	$$[\COMPOSE{C_1}{C_2}]^{\beta_0} (\beta) = \sum_{\beta' \in \mathcal{D}(\Sigma)}  [C_1]^{\beta_0} (\beta') \cdot [C_2]^{\beta'} (\beta) $$
	as well as 
	$$\llbracket \COMPOSE{C_1}{C_2} \rrbracket^{\beta_0} (\beta) = \sum_{\beta' \in \mathcal{D}(\Sigma)} \llbracket C_1 \rrbracket^{\beta_0} (\beta') \cdot \llbracket C_2 \rrbracket^{\beta'} (\beta) $$
	The first equation has already been argued for in the case $\COMPOSE{C_1}{C_2}$ of \cref{belief_accuracy}.	The second equation can be explained as follows:
	
	If we consider a particular terminal configuration $\langle \top, \beta, p\rangle$ in the computation tree of $\llbracket \COMPOSE{C_1}{C_2}\rrbracket^{\beta_0}$, then the path from the initial configuration $\langle \COMPOSE{C_1}{C_2}, \beta_0, 1\rangle$ to that leaf can be divided into two parts: We first execute $C_1$ independently of $C_2$ until $C_1$ terminates and we reach a configuration of the form $\langle C_2, \beta', q\rangle$, and then we execute $C_2$ until we finally reach the terminal configuration. This behavior is formalized in the equation above, but we additionally sum over all possible intermediate belief states $\beta'$ to get the total probability to terminate with the belief state $\beta$.
	The theorem now follows directly from the inductive hypothesis for $C_1$ and $C_2$ as well as the two equations.

	\subsection{\CHOOSE{x}{f}}
	We only have one reachable belief state, namely the belief state $\beta_0\statesubst{x}{f}$ with $[\CHOOSE{x}{f}]^{\beta_0}(\beta_0\statesubst{x}{f}) = 1$ and $\llbracket \CHOOSE{x}{f}\rrbracket^{\beta_0} (\beta_0\statesubst{x}{f}) = 1$. The theorem now follows directly.
	\subsection{\OBSERVEASSIGN{y}{x}}
	The observe statement can be split into two parts: A "pure" observe statement \OBSERVE{x} which only conditions the belief state via $\beta_0|_{x=c}$, and an assignment which updates the value of $y$ with the (now deterministic) value of $x$.
	It suffices to consider the first part, i.e. $\OBSERVE{x}$, as the second part can be proven fully analogous to the case for the assignment, because after the observe statement has been executed there is only a single possible value for $x$ in the belief state. 
	Afterwards the full observe statement follows from the proofs for $\OBSERVE{x}$, the assignment and sequential execution.
	
	If we now consider $\llbracket \OBSERVE{x}\rrbracket^{\beta_0}$, then we notice that the initial configuration $\langle \OBSERVE{x}, \beta_0, 1\rangle$ splits into multiple terminal configurations of the form $\langle \top, \beta_0|_{x=c}, \BMODEL{\beta_0}{x = c} \rangle$ for every possible value $c$ with $\BMODEL{\beta_0}{x = c} > 0$. Thus, all reachable, terminal belief states have the form $\beta_0|_{x=c}$ for some $c$. It now suffices to show: 
	$$ \llbracket \OBSERVE{x}\rrbracket^{\beta_0} (\beta_0|_{x=c}) = \BMODEL{\beta_0}{x = c} = [\OBSERVE{x}]^{\beta_0} (\beta_0|_{x=c}) $$ 
	The first equation follows directly from the transition rules of configurations. $[\OBSERVE{x}]^{\beta_0}$ is considered below. Notice that $[\OBSERVE{x}]^{\beta_0,\sigma}(\beta_0|_{x=c}) = 1$ iff $\sigma(x) = c$ and $0$ otherwise, this follows directly from the transition rules. Thus we get:
	\begin{align*}
		[\OBSERVE{x}]^{\beta_0}(\beta_0|_{x=c}) &= \sum\limits_{\sigma \in \Sigma}\beta_0(\sigma)\cdot [\OBSERVE{x}]^{\beta_0, \sigma}(\beta_0|_{x=c})\\ 
		&=_1 \sum\limits_{\sigma \in \Sigma}\beta_0(\sigma)\cdot\SMODEL{\sigma}{x=c}\\ 
		&= \BMODEL{\beta_0}{x = c}
	\end{align*}
	The equation $=_1$ holds as we have $[\OBSERVE{x}]^{\beta_0, \sigma}(\beta_0|_{x=c}) = 1$ whenever $\sigma(x)=c$ holds, and $0$ otherwise. The theorem now follows directly.
	\subsection{\INFER{p(P) \in i}{C_1}{C_2}}
	The proof is fully analogous to the if-statement.
\end{proof}

\section{Proof of the Soundness Theorem}\label{Soundness_proof}
In this section we formally prove \Cref{soundness}. 
We will not prove the theorem directly, instead we will prove the following theorem, which uses an alternative definition of the operational semantics without a variable assignment.
The alternative semantics simplify the proof, and we present them in more detail in \Cref{semantics_full}.
The original theorem then follows from the statement below and \Cref{OperationalCorrectnessTheorem}.
We now state the theorem which we will prove:

\begin{theorem}[Adapted Soundness]
Let $C$ be a program, $\beta$ an initial belief state and $F$ a predicate on belief states. Then we have:
$$ \wp{C}{F}\, (\beta_0) \eeq \sum_{\beta \in \mathcal{D}(\Sigma)} \llbracket C \rrbracket^{\beta_0}(\beta) \cdot F (\beta)$$
\end{theorem}
\begin{proof}
	We will prove the theorem by structural induction over $C$, giving us the following cases:
	\subsection{\SKIP}
	We have $\llbracket \SKIP \rrbracket^{\beta_0}(\beta_0) = 1$ and $\llbracket \SKIP \rrbracket^{\beta_0}(\beta) = 0$ for all other belief states $\beta$, thus we have 
	$$\wp{\SKIP}{F}\, (\beta_0) = F (\beta_0) = \llbracket \SKIP \rrbracket^{\beta_0} (\beta_0) \cdot F(\beta_0) = \sum_{\beta \in \mathcal{D}(\Sigma)} \llbracket \SKIP \rrbracket^{\beta_0}(\beta) \cdot F (\beta)$$
	
	\subsection{\ASSIGN{x}{E}}
	We have $\llbracket \ASSIGN{x}{E}\rrbracket^{\beta_0}(\beta_0\statesubst{x}{E}) = 1$ and $\llbracket \ASSIGN{x}{E} \rrbracket^{\beta_0}(\beta) = 0$ for all other belief states $\beta$, thus we can conclude:
	\begin{align*}
		\wp{\ASSIGN{x}{E}}{F}\, (\beta_0) &= F\statesubst{x}{E} (\beta_0)\\ 
		&= F (\beta_0\statesubst{x}{E})\\
		&= \llbracket  \ASSIGN{x}{E}\rrbracket^{\beta_0} (\beta_0\statesubst{x}{E}) \cdot F(\beta_0\statesubst{x}{E})\\ 
		&= \sum_{\beta \in \mathcal{D}(\Sigma)} \llbracket \ASSIGN{x}{E} \rrbracket^{\beta_0}(\beta) \cdot F (\beta)
	\end{align*}
	\subsection{\ITE{P_o}{C_1}{C_2}}
	$P_o$ is an observable proposition, which means we have either $Pr(P_o) (\beta_0)= \BMODEL{\beta_0}{P_o} = 1$ or $Pr(P_o) (\beta_0) = \BMODEL{\beta_0}{P_o} = 0$, as we only consider consistent belief states. Now assume w.l.o.g. that $\BMODEL{\beta_0}{P_o} = 1$ holds. 
	
	Thus we have the transition $\langle \ITE{P_o}{C_1}{C_2}, \beta, 1 \rangle \triangleright \langle C_1,\beta, 1 \rangle$. From that follows that $\llbracket\IFSYMBOL\dots\rrbracket^{\beta_0} = \llbracket C_1\rrbracket^{\beta_0}$ holds. As $C_1$ is structurally smaller than the if-statement, we may apply the inductive hypothesis and the theorem can be proven as follows:
	\begin{align*}
		&\wp{\IFSYMBOL\dots}{F} (\beta_0) \\
		&= Pr(P_o) \cdot \wp{C_1}{F}\, (\beta_0) + (1 - Pr(P_o)) \cdot \wp{C_2}{F}\, (\beta_0)\\
		&= \BMODEL{\beta_0}{P} \cdot \wp{C_1}{F}\, (\beta_0) + (1 - \BMODEL{\beta_0}{P}) \cdot \wp{C_2}{F}\, (\beta_0)\\
		&= \wp{C_1}{F}\, (\beta_0)\\
		&=_{IH} \sum_{\beta \in \mathcal{D}(\Sigma)} \llbracket C_1 \rrbracket^{\beta_0}(\beta) \cdot F (\beta)\\ 
		&= \sum_{\beta \in \mathcal{D}(\Sigma)} \llbracket \IFSYMBOL\dots \rrbracket^{\beta_0}(\beta) \cdot F (\beta) 
	\end{align*}
	The case for $(\BMODEL{\beta}{P}) = 0$ is proven fully analogously.
	\subsection{\WHILEDO{P_o}{C}}
	We will first consider the theorem for bounded loops, defined as 
	$$\BOUNDEDWHILE{n}{P_o}{C} \coloneqq \begin{cases} \diverge &, n = 0\\
		\ITE{P_o}{\COMPOSE{C}{\BOUNDEDWHILE{n-1}{P_o}{C}}}{\SKIP} &, n>0
	\end{cases}$$
	where we require $\wp{\diverge}{F}\, \beta = 0$. Afterwards, we can then generalize to unbounded while loops as we have: 
	$$\lim\limits_{n \rightarrow \infty} \wp{\BOUNDEDWHILE{n}{P}{C}}{F} =  \wp{\WHILEDO{P}{C}}{F}  $$
	and 
	$$\lim\limits_{n \rightarrow \infty} \llbracket \BOUNDEDWHILE{n}{P}{C}\rrbracket^{\beta_0} (\beta) = \llbracket \WHILEDO{P}{C}\rrbracket^{\beta_0} (\beta)$$
	A full proof for the first equation is given in \Cref{wpConvergence}. The second equation has already been proven in the while loop case of \Cref{Tree_proof}. Furthermore we will show the following equation for all $k\in \mathbb{N}$:
	$$ \wp{\BOUNDEDWHILE{k}{P}{C}}{F} \, (\beta_0) = \sum_{\beta \in \mathcal{D}(\Sigma)} \llbracket \BOUNDEDWHILE{k}{P}{C} \rrbracket^{\beta_0}(\beta) \cdot F (\beta)$$
	The soundness theorem follows directly from the three equations above. We will now prove the third equation for all $k\in \mathbb{N}$.
	\begin{proof}
		We will prove this fact by induction on $k$:
		\begin{itemize}
			\item [$k=0$:] We have 
			\begin{align*}
				\wp{\BOUNDEDWHILE{0}{P}{C}}{F}\, (\beta_0) &= \wp{\diverge}{F}\, (\beta_0)\\
				&= 0\\
				&= \sum_{\beta \in \mathcal{D}(\Sigma)} \llbracket \diverge \rrbracket^{\beta_0}(\beta) \cdot F (\beta)\\
				&= \sum_{\beta \in \mathcal{D}(\Sigma)} \llbracket \BOUNDEDWHILE{0}{P}{C} \rrbracket^{\beta_0}(\beta) \cdot F (\beta)
			\end{align*}
			\item [$k > 0$:] We make a case distinction on whether we have $\BMODEL{\beta_0}{P} = 1$ or $\BMODEL{\beta_0}{P} = 0$:
			\begin{itemize}
				\item $\BMODEL{\beta_0}{P} = 1$. By the definition of $\wpsymbol$ and the semantics of the if-statement we only need to show 
				\begin{align*}
					&\wp{\COMPOSE{C}{\BOUNDEDWHILE{k-1}{P}{C}}}{F}\, (\beta_0) \\
					&= \sum_{\beta \in \mathcal{D}(\Sigma)} \llbracket \COMPOSE{C}{\BOUNDEDWHILE{k-1}{P}{C}} \rrbracket^{\beta_0}(\beta) \cdot F (\beta)
				\end{align*}
				Which follows from the inductive hypothesis for $k-1$ and the case for $\COMPOSE{C_1}{C_2}$
				\item $\BMODEL{\beta_0}{P} = 0$. By the definition of $\wpsymbol$ and the semantics of the if-statement we only need to show 
				$$\wp{\SKIP}{F}\, (\beta_0) = \sum_{\beta \in \mathcal{D}(\Sigma)} \llbracket \SKIP \rrbracket^{\beta_0}(\beta) \cdot F (\beta)$$
				Which follows from the previously shown case for $skip$
			\end{itemize}
		\end{itemize}
	\end{proof} 
	\subsection{$\COMPOSE{C_1}{C_2}$}
	By applying the inductive hypothesis using $C_1$ and $C_2$ to the definition of $wp(\COMPOSE{C_1}{C_2}, F)$, we get:
	\begin{align*}
		\wp{\COMPOSE{C_1}{C_2}}{F}(\beta_0) &= \wp{C_1}{\wp{C_2}{F}}(\beta_0)\\ 
		&=_{IH} \sum_{\beta' \in \mathcal{D}(\Sigma)} \llbracket C_1 \rrbracket^{\beta_0}(\beta') \cdot \wp{C_2}{F}(\beta')\\ 
		&=_{IH} \sum_{\beta' \in \mathcal{D}(\Sigma)} \llbracket C_1 \rrbracket^{\beta_0}(\beta') \cdot \left(\sum_{\beta'' \in \mathcal{D}(\Sigma)} \llbracket C_2 \rrbracket^{\beta'}(\beta'') \cdot F (\beta'')\right)\\
		&= \sum_{\beta' \in \mathcal{D}(\Sigma)} \sum_{\beta'' \in \mathcal{D}(\Sigma)} \llbracket C_1 \rrbracket^{\beta_0}(\beta') \cdot \llbracket C_2 \rrbracket^{\beta'}(\beta'') \cdot F (\beta'')\\
		&= \sum_{\beta'' \in \mathcal{D}(\Sigma)} \sum_{\beta' \in \mathcal{D}(\Sigma)} \llbracket C_1 \rrbracket^{\beta_0}(\beta') \cdot \llbracket C_2 \rrbracket^{\beta'}(\beta'') \cdot F (\beta'')\\
		&= \sum_{\beta'' \in \mathcal{D}(\Sigma)} F(\beta'') \cdot \sum_{\beta' \in \mathcal{D}(\Sigma)} \llbracket C_1 \rrbracket^{\beta_0}(\beta') \cdot \llbracket C_2 \rrbracket^{\beta'}(\beta'') \\
		&= \sum_{\beta'' \in \mathcal{D}(\Sigma)} F(\beta'') \cdot  \llbracket \COMPOSE{C_1}{C_2} \rrbracket^{\beta_0}(\beta'')
	\end{align*}
	The first three equations follow from the definition of $\wpsymbol$ and the inductive hypothesis using $C_1$ and $C_2$. The next equations only modify the expression based on common rules of arithmetic. The last equation follows as the execution of $\COMPOSE{C_1}{C_2}$ can be understood as first executing $C_1$, terminating in an intermediate belief state $\beta'$ and then executing $C_2$. A more thorough justification of this is given in the case for $\COMPOSE{C_1}{C_2}$ of \Cref{Tree_proof}. The final expression amounts to the expression from the soundness theorem up to renaming of $\beta''$.

	\subsection{\CHOOSE{x}{f}}
	This case is analogous to $\ASSIGN{x}{E}$: 
	We have $\llbracket \CHOOSE{x}{f} \rrbracket^{\beta_0}(\beta_0\statesubst{x}{f}) = 1$ and $\llbracket \CHOOSE{x}{f} \rrbracket^{\beta_0}(\beta) = 0$ for all other belief states $\beta$, thus we have 
	\begin{align*}
		\wp{\CHOOSE{x}{f}}{F}\, (\beta_0) &= F\statesubst{x}{f} (\beta_0)\\ 
		&= F (\beta_0\statesubst{x}{f})\\ 
		&= \llbracket \CHOOSE{x}{f}\rrbracket^{\beta_0} (\beta_0\statesubst{x}{f}) \cdot F(\beta_0\statesubst{x}{f})\\ 
		&= \sum_{\beta \in \mathcal{D}(\Sigma)} \llbracket \CHOOSE{x}{f} \rrbracket^{\beta_0}(\beta) \cdot F (\beta)
	\end{align*}

	\subsection{\OBSERVEASSIGN{y}{x}}
	The observe statement can be split into two parts: A "pure" observe statement \OBSERVE{x} which only conditions the belief state via $\beta_0|_{x=c}$, and an assignment which updates the value of $y$ with the (now deterministic) value of $x$.
	It suffices to consider the first part, i.e. $\OBSERVE{x}$, as the second part can be proven fully analogous to the case for the assignment, because after the observe statement has been executed there is only a single possible value for $x$ in the belief state.
	Afterwards the full observe statement follows from the proofs for $\OBSERVE{x}$, the assignment and sequential execution.
	
	We have $\llbracket \OBSERVE{x} \rrbracket^{\beta_0} (\beta_0|_{(x=c)}) = \BMODEL{\beta_0}{x = c} = Pr(x=c)(\beta_0)$ for any value $c \in \mathbb{N}$, as well as $\llbracket \OBSERVE{x}\rrbracket^{\beta_0} (\beta) = 0$ for all other belief states $\beta$. Thus we get:
	\begin{align*}
		\wp{\OBSERVE{x}}{F} (\beta_0) &= \sum_{c \in \mathbb{N}} Pr(x = c) \cdot F|_{x = c}(\beta_0) \\
		&= \sum_{c \in \mathbb{N}} \BMODEL{\beta_0}{x = c} \cdot F(\beta_0|_{x = c}) \\
		&= \sum_{c \in \mathbb{N}} \llbracket \OBSERVE{x}\rrbracket^{\beta_0}(\beta_0|_{(x=c)}) \cdot F(\beta_0|_{x = c})\\
		&= \sum_{\beta \in \mathcal{D}(\Sigma)} \llbracket \OBSERVE{x}\rrbracket^{\beta_0}(\beta) \cdot F(\beta)
	\end{align*}  
	\subsection{\INFER{p(P) \in i}{C_1}{C_2}}
	The proof for the infer statement is fully analogous to the if-statement
\end{proof}

\section{Convergence of $\wpsymbol$ for Bounded Loops}\label{wpConvergence}
To show:
$$\lim\limits_{n \rightarrow \infty} \wp{\BOUNDEDWHILE{n}{P}{C}}{F}(\beta_0) =  \wp{\WHILEDO{P}{C}}{F} (\beta_0) $$
In order to prove this, we will apply the following theorem:
\begin{theorem}[Theorem 3 in \cite{lassez1982fixed}]
	Every scott-continuous function $F$ over a closed partial order has a least fixed point which is $\sup_{n \in \mathbb{N}} F^n(\bot)$
\end{theorem}
In \Cref{Soundness_proof}, we defined the $\wpsymbol$ for bounded while loops in a purely syntactic manner by reducing to the rules for the if-statement and diverge. However, when we consider the characteristic function $\Phi_F$ of a while loop with regard to the postexpectation $F$, i.e.
$$ \Phi_F (X) = Pr(P) \cdot \wp{C}{X} + (1- Pr(P)) \cdot F$$
then we can see that the syntactic definition of $\wpsymbol$ for bounded loops is directly equivalent to repeatedly applying the characteristic function to the constant $0$ predicate, i.e.\ $ \wp{\BOUNDEDWHILE{n}{P}{C}}{F} = \Phi_F^n(0)$. If we can now show that $\wpsymbol$ is scott-continuous, then we can use the theorem above to show:
$$\wp{\WHILEDO{P}{C}}{F} = \lfp\Phi_F = \sup\limits_{n \in \mathbb{N}} \Phi_F^n(0) =  \sup\limits_{n \in \mathbb{N}}  \wp{\BOUNDEDWHILE{n}{P}{C}}{F}$$ 
Furthermore, as $ \wp{\BOUNDEDWHILE{n}{P}{C}}{F}$ is monotonically increasing with $n$, we may exchange the supremum for the limit and we would arrive at the desired theorem. The rest of this section is used to introduce the necessary definitions and to prove that $\wpsymbol$ is scott-continuous in order to apply the theorem.
\begin{definition}
	A set $D \neq \{\}$ is called directed, iff $\forall x,y \in D\colon\exists z\in D\colon z\geq x \land z \geq y$
\end{definition}
\begin{definition}\label{cpo}
	The set $S$ is a complete partial order (cpo) iff for every directed subset $D\subseteq S$ the supremum of $D$ exists and lies in $S$
\end{definition}
The set of all possible predicates $F\colon \mathcal{D}(\Sigma) \rightarrow \PosRealsInf$ is directed as we could simply take the point-wise maximum of two predicates. Furthermore, the set of our predicates is a cpo as we included $\infty$.
\begin{definition}
	A function $F\colon S\rightarrow S'$ with $S,S'$ both being cpos is scott-continuous iff
	\begin{itemize}
		\item the image $F(D)$ for every directed set $D$ is directed
		\item $F(\sup\,S) = \sup\,F(S)$
	\end{itemize}
\end{definition}
We now show that for every program $C$, the function $\wpC{C}$ is scott-continuous. 
We will prove both points at the same time by structural induction on $C$. 
The proof is very similar to the proof of the analogous theorem in \cite{gretz2014operational}, where the only notable difference are the additional cases for the sample, observe and infer statement.
Nonetheless, we will present the full proof.
For the first point we will show that $\wpsymbol$ is monotonous, i.e.\ we have 
$$f\geq g \Rightarrow \wp{C}{f} \geq \wp C g$$
The directedness of $wp(C,D)$ then follows directly from the monotonicity of $\wpsymbol$. The proof of monotonicity will be omitted in all cases except the while loop, as it is trivial and follows directly by applying the definition of $\wpsymbol$ and the inductive hypothesis.

For the second point, we have to show for all directed sets of predicates $D$:
$$ \wp{C} {\sup_{f \in D}\,f }= \sup\limits_{f \in D}\,\wp{C}{f}$$
\subsection{\SKIP}
We have 
$$ \wp{\SKIP}{\sup\limits_{f \in D}\,f} = \sup\limits_{f \in D}\,f = \sup\limits_{f \in D}\,\wp{\SKIP}{f} $$
\subsection{\ASSIGN{x}{E}}
We have 
\begin{align*}
	\wp{\ASSIGN{x}{E}}{\sup\limits_{f \in D}\,f}
	&= (\sup\limits_{f \in D}\,f)\, \statesubst{x}{E}\\
	&=_1 \sup\limits_{f \in D}\, (f\statesubst{x}{E})\\
	&= \sup\limits_{f \in D}\,\wp{\ASSIGN{x}{E}}{f}
\end{align*}
For $=_1$ we will show that if we have $f_1 = \sup_{f \in D}$, then we also have $f_1\statesubst{x}{E} = \sup_{f \in D} f[x \mapsto E]$. In particular, we will show that for every $f_2 \in D$ that $f_1\statesubst{x}{E}$ is point-wise greater than $f_2\statesubst{x}{E}$:
\begin{align*}
	f_1\statesubst{x}{E} (\beta) &= f_1 (\beta [x \mapsto E])\\ 
	&\geq f_2 (\beta [x \mapsto E]) \\
	&= f_2\statesubst{x}{E} (\beta)
\end{align*}
$\geq$ follows as $f_1$ is point-wise greater than $f_2$ on all belief states, which includes $\beta [x \mapsto E]$
\subsection{\CHOOSE{x}{f}}
Fully analogous to the previous case
\subsection{\OBSERVEASSIGN{y}{x}}
The observe statement can be split into two parts: A "pure" observe statement \OBSERVE{x} which only conditions the belief state via $\beta_0|_{x=c}$, and an assignment which updates the value of $y$ with the (now deterministic) value of $x$.
It suffices to consider the first part, i.e. $\OBSERVE{x}$, as the second part can be proven fully analogous to the case for the assignment, because after the observe statement has been executed there is only a single possible value for $x$ in the belief state.
Afterwards the full observe statement follows from the proofs for $\OBSERVE{x}$, the assignment and sequential execution.
We thus have:
\begin{align*}
	\wp{\OBSERVE{x}}{\sup\limits_{f \in D}\,f}
	&= \sum_{c \in \mathbb{N}}Pr(x = c)\cdot (\sup\limits_{f \in D}\,f)|_{x=c}\\
	&=_1 \sum_{c \in \mathbb{N}}Pr(x = c)\cdot  \sup\limits_{f \in D}\,(f|_{x=c})\\
	&\geq_2 \sup\limits_{f \in D} \, \sum_{c \in \mathbb{N}}Pr(x = c)\cdot  \,f|_{x=c}\\
	&= \sup\limits_{f \in D}\,\wp{\OBSERVE{x}}{f}
\end{align*}
Here, $=_1$ follows similar to the case for $\ASSIGN{x}{E}$: It does not matter if we take the supremum before or after the recursive call and we can move the supremum outwards. 
We will now show that $\geq_2$ can be strengthened to an equality by showing that $>$ leads to a contradiction, and the proof for this can be understood as a generalization of the proof we give for the if-statement. Assume that we have
$$ \sum_{c \in \mathbb{N}}Pr(x = c)\cdot  \sup\limits_{f \in D}\,(f|_{x=c})
> \sup\limits_{f \in D} \, \sum_{c \in \mathbb{N}}Pr(x = c)\cdot  \,f|_{x=c}$$
Then by the definition of the supremum as the smallest upper bound, there must be $g_c \in D$ for $c\in \mathbb{N}$ so that we have
$$ \sum_{c \in \mathbb{N}}Pr(x = c)\cdot (g_c|_{x=c})
> \sup\limits_{f \in D} \, \sum_{c \in \mathbb{N}}Pr(x = c)\cdot  \,f|_{x=c}$$
If there would be no such $g_c$, then we immediately reach a contradiction. 
Now consider an arbitrary index $d \in \mathbb{N}$.
By the definition of $>$, The inequation above has to hold for all $\beta \in \mathcal{D}(\Sigma)$, and thus the inequation also holds for any $\beta_d$ with $Pr(x = d)\, \beta_d = 1$ and $Pr(x = c)\, \beta_d = 0$ for all $d\neq c, c\in \mathbb{N}$. By considering the evaluation of both terms on $\beta_d$ we can "filter out" a single summand, and we now get the following contradiction:
\begin{align*}
	\left( \sup\limits_{f \in D} \, \sum_{c \in \mathbb{N}}Pr(x = c)\cdot  \,f|_{x=c}\right) \,(\beta_d) 
	&< \left(\sum_{c \in \mathbb{N}}Pr(x = c)\cdot (g_c|_{x=c}) \right) \,(\beta_d)\\
	&= (g_d|_{x=d}) \,(\beta_d)\\
	&= \left(\sum_{c \in \mathbb{N}}Pr(x = c)\cdot (g_d|_{x=c})\right) \,(\beta_d)\\
	&\leq \left( \sup\limits_{f \in D} \, \sum_{c \in \mathbb{N}}Pr(x = c)\cdot  \,f|_{x=c}\right) \,(\beta_d)
\end{align*}

\subsection{\ITE{P}{C_1}{C_2}}
We have:
\begin{align*}
	&\wp{\ITE{P}{C_1}{C_2}}{\sup\limits_{f \in D}\,f}\\
	&= Pr(P) \cdot \wp{C_1}{\sup\limits_{f \in D}\,f} + (1-Pr(P)) \cdot \wp{C_2}{\sup\limits_{f \in D}\,f}\\
	&=_{IH} (\sup\limits_{f \in D}\, Pr(P) \cdot \wp{C_1}{f}) + (\sup\limits_{f \in D}\, (1-Pr(P)) \cdot \wp{C_2}{f})\\
	&\geq_1 \sup\limits_{f \in D}\, (Pr(P) \cdot \wp{C_1}{f}+ (1- Pr(P)) \cdot \wp{C_1}{f}) \\
	&= \sup\limits_{f \in D}\,\wp{\ITE{P}{C_1}{C_2}}{f}
\end{align*}
We can strengthen $\geq_1$ to an equality by showing that the we cannot have $>$. Assume we have:
\begin{align*}
	&( \sup\limits_{f \in D}\,Pr(P) \cdot \wp{C_1}{f}) + ( \sup\limits_{f \in D}\, (1-Pr(P)) \cdot \wp{C_2}{f})\\
	&> \sup\limits_{f \in D}\, (Pr(P) \cdot \wp{C_1}{f} + (1-Pr(P)) \cdot \wp{C_2}{f})
\end{align*}
Then, by the definition of supremum as the smallest upper bound there must be $g_1,g_2 \in D$ with:
\begin{align*}
	&(Pr(P) \cdot \wp{C_1}{g_1}) + ( (1-Pr(P)) \cdot \wp{C_1}{g_2})\\ 
	&> \sup\limits_{f \in D}\, (Pr(P) \cdot \wp{C_1}{f} + (1-Pr(P)) \cdot \wp{C_2}{f})
\end{align*}
Then we have a contradiction by:
\begin{align*}
	&\sup\limits_{f \in D}\, (Pr(P) \cdot \wp{C_1}{f} + (1 - Pr(P)) \cdot \wp{C_2}{f}) \\
	&< (Pr(P) \cdot \wp{C_1}{g_1}) + ( (1 - Pr(P)) \cdot \wp{C_2}{g_2}) \\
	&\leq_2 (Pr(P) \cdot \wp{C_1}{h}) + ( (1 - Pr(P)) \cdot \wp{C_2}{h}) \\
	&\leq \sup\limits_{f \in D} Pr(P) \cdot \wp{C_1}{f} + (1 - Pr(P)) \cdot \wp{C_2}{f}
\end{align*}
The existence an $h\in D$ with $h\geq g_1$ and $h\geq g_2$ follows from the directedness of $D$, and $\leq_2$ follows from the monotonicity of $\wpsymbol$ via the inductive hypothesis.
\subsection{\INFER{p(P)\in i}{C_1}{C_2}}
Analogous to the case for the if-statement
\subsection{$\COMPOSE{C_1}{C_2}$}
By applying the inductive hypothesis twice we get:
\begin{align*}
	\wp{\COMPOSE{C_1}{C_2}}{\sup\limits_{f \in D}\,f}
	&= \wp{C_1} {\wp{C_2} {\sup\limits_{f \in D}\,f}}\\
	&= \wp{C_1} {\sup\limits_{f \in D}\, \wp{C_2}{F}}\\
	&= \sup\limits_{f \in D}\,\wp{C_1} {\wp {C_2} f}\\
	&= \sup\limits_{f \in D}\,\wp{\COMPOSE{C_1}{C_2}}{f}
\end{align*}
\subsection{\WHILEDO{P}{C}}
By definition, the weakest preexpectation with regard to postexpectation $f$ for the loop $while (P)\, C$ is defined as the least fixpoint of a function 
$$\Phi_f (X) = Pr(P) \cdot \wp{C}{X} + (1-Pr(P)) \cdot f$$
Note that $\Phi_f (X)$ only depends on $C$ which is structurally smaller than the while loop as a whole. We can thus use the inductive hypothesis for $C$ to show that $\wp{C}{X}$ is scott-continuous in $X$. It now follows directly that $\Phi_f(X)$ is also scott-continous. We thus have:
$$\lfp\Phi_f = \sup\limits_{n\in \mathbb{N}} \Phi^n_f(0)$$
We will now use this fact to show both monotonicity and the ability to move the supremum outwards:

\paragraph{Monotonicity:}\label{mono} Assume $f\leq g$. We will first show that we have $\Phi^n_f(0) \leq \Phi^n_g(0)$ for all $n \in \mathbb{N}$. Proof by induction on $n$:
\begin{itemize}
	\item $n=0$: We have $\Phi^0_f(0) = 0 \leq 0 = \Phi^0_g(0)$
	\item $n \rightarrow n+1$: We have
	\begin{align*}
		\Phi^{n+1}_f(0) &= Pr(P) \cdot \wp{C} {\Phi^n_f(0)} + (1 - Pr(P)) \cdot f\\
		&\leq_1 Pr(P) \cdot \wp{C} {\Phi^n_g(0)} + (1 - Pr(P)) \cdot f \\
		&\leq_2 Pr(P) \cdot \wp{C} {\Phi^n_g(0)} + (1 - Pr(P)) \cdot g \\
		&= \Phi^{n+1}_g(0)
	\end{align*}
	In the equations above, $\leq_1$ follows from the inductive hypothesis for $n$. $\leq_2$ follows from the assumption $f\leq g$. It now follows:
	\begin{align*}
		\lfpX \Phi_f(X) &= \sup\limits_{n\in \mathbb{N}} \Phi^n_f(0)\\
		&\leq \sup\limits_{n\in \mathbb{N}} \Phi^n_g(0)\\
		&=\lfpX \Phi_g(X)
	\end{align*}
\end{itemize}
\paragraph{Continuity:} For the sake of notation we will use $\Phi_{sup}(X) \coloneqq \Phi_{\sup\limits_{f\in D}f} (X)$
\begin{align*}
	\wp{\WHILEDO{P}{C}}{\sup_{f\in D}f} &= \sup\limits_{n\in \mathbb{N}} \Phi^n_{sup}(0)\\
	&=_1 \sup\limits_{n\in \mathbb{N}} (\sup\limits_{f\in D}\,\Phi_f)^n(0)\\
	&=_2 \sup\limits_{n\in \mathbb{N}} \sup\limits_{f\in D}\,\Phi_f^n(0)\\
	&=_3 \sup\limits_{f\in D} \sup\limits_{n\in \mathbb{N}} \Phi^n_f(0)\\
	&= \sup\limits_{f\in D} \wp{\WHILEDO{P}{C}}{f}
\end{align*}
We will now explain the marked equations:
\begin{itemize}
	\item
	We have $=_1$ as it does not matter whether we add $Pr(P) \cdot \wp{C}{X}$ and multiply with $1 - Pr(P)$ before constructing the supremum or afterwards, as those terms do not depend on $f$
	\begin{align*}
		\Phi_{sup} (X) &= Pr(P) \cdot \wp{C}{X} + (1 - Pr(P)) \cdot \sup\limits_{f\in D} f\\
		&= \sup\limits_{f\in D}(Pr(P) \cdot \wp{C}{X} + (1 -Pr(P)) \cdot f)\\ 
		&= \sup\limits_{f\in D} \Phi_f(X)
	\end{align*}
	
	\item
	We have $=_2$ as we can show $(\sup_{f\in D}\,\Phi_f)^n(0) = \sup_{f\in D}\,\Phi_f^n(0)$ holds for every $n\in \mathbb{N}$. Proof by induction on $n$:
	\begin{itemize}
		\item $n = 0$: We have $(\sup\limits_{f\in D}\,\Phi_f)^0(0) = 0 = \sup\limits_{f\in D}\,\Phi_f^0(0)$
		\item $n \rightarrow n+1$: We have
		\begin{align*}
			(\sup\limits_{f\in D}\,\Phi_f)^{n+1}(0) &= \sup\limits_{f\in D}\,\Phi_f((\sup\limits_{f\in D}\,\Phi_f)^{n}(0))\\
			&=_{a} \sup\limits_{f\in D}\,\Phi_f(\sup\limits_{f\in D}\,\Phi_f^{n}(0))\\
			&=_{b} \sup\limits_{f\in D} \sup\limits_{f\in D}\,\Phi_f(\,\Phi_f^{n}(0))\\
			&= \sup\limits_{f\in D}\,\Phi_f(\Phi_f^{n}(0)\\
			&= \sup\limits_{f\in D}\,\Phi_f^{n+1}(0)
		\end{align*}
		$=_a$ holds by the inductive hypothesis. $=_{b}$ holds as $D$ is directed and $\Phi_f$ is scott-continuous, we are thus able to pull the supremum outwards past the application of $\Phi_f$.
	\end{itemize}
	
	\item
	We prove $=_3$ by establishing inequality in both directions:
	\begin{align*}
		&\forall n\in \mathbb{N}, \;\forall f \in D: &\Phi_f^n (0)\leq \sup\limits_{g\in D} \Phi_g^n (0)\\
		&\Rightarrow\forall f \in D: &\sup\limits_{n\in \mathbb{N}} \Phi_f^n (0)\leq \sup\limits_{n\in \mathbb{N}} \sup\limits_{g\in D} \Phi_g^n (0)\\
		&\Rightarrow &\sup\limits_{f\in D}\sup\limits_{n\in \mathbb{N}} \Phi_f^n (0)\leq \sup\limits_{n\in \mathbb{N}} \sup\limits_{g\in D} \Phi_g^n (0)\\
	\end{align*}
	The last inequation follows from the definition of the suprema, as it is defined to be the least upper bound of the set $M \coloneqq \{\sup_{n\in \mathbb{N}} \Phi_f^n (0) | f\in D\}$. Note that the second inequation establishes $\sup_{n\in \mathbb{N}} \sup_{g\in D} \Phi_g^n (0)$ as an upper bound of the set $M$.
	The existence of these suprema follows as $D$ is a directed set and $\Phi$ is scott-continuous by the inductive hypothesis. In a fully analogous manner we can also show:
	$$ \sup\limits_{f\in D}\sup\limits_{n\in \mathbb{N}} \Phi_f^n (0)\geq \sup\limits_{n\in \mathbb{N}} \sup\limits_{g\in D} \Phi_g^n (0)$$
\end{itemize}

\section{Linearity and Independence}\label{app:properties}

\subsection{Formal Introduction of Linearity and Independence}
In the following, we will discuss two useful properties of $\wpsymbol$, which allows us to calculate $\wpsymbol$ more easily. 
The first property -- \textit{linearity} -- enables \emph{decomposing} predicates, making $\wpsymbol$ computations amenable to parallelization. 
Formally, we describe linearity as follows:
\begin{theorem}[Linearity of $\wpsymbol$]\label{Linearity}
	Let $C \in \PBLIMP$, $\alpha \in \mathbb{R}_{{}\geq 0}$, and $F, G$ be predicates. 
	Then,%
	\begin{align*}
		\wp{C}{\alpha\cdot F + G} \qeq \alpha \cdot \wp{C}{F} \pplus \wp{C}{G}~.
	\end{align*}%
\end{theorem} 
\begin{proof}
	By structural induction on $C$, see \Cref{LinearityProof}. \qed
\end{proof}%
\noindent%
The second property -- \textit{independence} -- enables to \enquote{skip} calculating $\wp{C}{F}$ for a predicate if we can show that the program $C$ does not affect the value of~$F$. 
We formalize this idea through the function $\mathit{mod}(C)$, which maps loop-free belief programs to the set of variables affected by $C$. 
We can define $\mathit{mod}$ as:%

\begin{adjustbox}{max width = \textwidth}
	\begin{minipage}{\textwidth}
		\begin{align*} 
			\mathit{mod}(C) \coloneqq \begin{cases}
				\mathit{mod}(C_1) \cup \mathit{mod}(C_2),	&\textnormal{if } C \eeq \COMPOSE{C_1}{C_2} \quad\textnormal{or}\quad  C \eeq \texttt{i$\ldots$) \{}C_1\texttt{\} else \{} C_2 \texttt{\}}\\
				x, 									&\textnormal{if } C\in \{\ASSIGN{x}{E},\,\CHOOSE{x}{f} \}\\
				\Varsu \cup \{y\}, 						&\textnormal{if } C \eeq \OBSERVEASSIGN{y}{x}
			\end{cases}\\
		\end{align*}
		\end{minipage}
	\end{adjustbox}%
	
	\noindent%
	Note that we statically over-approximate the set of modified variables, in particular for observe statements by using the set of all unobservable variables. 
	A more accurate set $\mathit{mod}(\OBSERVEASSIGN{y}{x})$ can be determined by considering the dependencies between unobservable variables which are created by sample statements. 
	We now formalize the independence  property as follows:%
	\begin{theorem}[Independence of $\wpsymbol$]\label{Unaffected}
		If none of the variables occurring freely in proposition~$P$ are in $\mathit{mod}(C)$, for $C$ loop-free, then we have 
		\begin{enumerate}
			\item $\wp{C}{Pr(P)} \eeq Pr(P)$
			\item $\wp{C}{[Pr(P) \in i]} \eeq [Pr(P) \in i]$
		\end{enumerate}
	\end{theorem}
	\begin{proof}
		By structural induction on $C$, see \Cref{UnaffectedProof}. 
		The restriction that $C$ is loop-free is required as the likelihood to diverge is always factored in with a weight of $0$, and would falsify the theorem for divergent $C$'s.%
		\qed
	\end{proof}%
	\noindent%
	\Cref{Unaffected} allows us to argue that the conditions of if-statements are not affected by observe or sample statements, as those conditions are restricted to observable variables. 
	Proving \Cref{Unaffected} requires insight into the evaluation of the predicate, and thus the theorem cannot be extended to arbitrary predicates. 
	Nonetheless, it is possible to prove an analogous theorem for the upcoming expected value predicate $Ex(E)$ from \Cref{GrammarFigure}.

\subsection{Proof of \Cref{Linearity}}\label{LinearityProof}
To show:
Let $C$ be a belief program, $\alpha \in \mathbb{R}_{{}\geq 0}$ and $F, G$ be predicates. Then we have
$$ \wp C {\alpha\cdot F + G} = \alpha \cdot \wp{C}{F} + \wp{C}{G}$$
\begin{proof}
	Proof by structural induction on $C$:
	\begin{itemize}
		\item $\SKIP$:
		\begin{align*}
			\wp{\SKIP}{\alpha\cdot F + G} &= \alpha\cdot F + G\\ 
			&= \alpha\cdot \wp \SKIP F + \wp \SKIP  G
		\end{align*}
		\item $\ASSIGN{x}{E}$:
		\begin{align*}
			\wp{\ASSIGN{x}{E}}{ \alpha\cdot F + G}\, (\beta) &= (\alpha\cdot F + G)\statesubst{x}{E} (\beta)\\ 
			&= (\alpha\cdot F + G) (\beta\statesubst{x}{E})\\ 
			&= \alpha\cdot F \,(\beta\statesubst{x}{E})+ G\,(\beta\statesubst{x}{E}) \\
			&= \alpha\cdot \wp {\ASSIGN{x}{E}} F \, (\beta) + \wp{\ASSIGN{x}{E}} G\,(\beta)
		\end{align*}
		\item \ITE{P_o}{C_1}{C_2}:
		\begin{align*}
			&\wp{\IFSYMBOL\dots} {\alpha\cdot F + G}\\
			&= Pr(P_o) \cdot \wp{C_1} {\alpha\cdot F + G} + (1 - Pr(P_o)) \cdot \wp{C_2}{\alpha\cdot F + G}\\
			&=_{IH} Pr(P_o) \cdot\left(\alpha \cdot \wp{C_1}{F} + \wp{C_1}{G}\right)\\
			&\quad\,+ (1 - Pr(P_o)) \cdot \left(\alpha \cdot \wp{C_2}{F} + \wp{C_2}{G}\right)\\
			&= \alpha \cdot Pr(P_o) \cdot \wp{C_1}{F} + Pr(P_o) \cdot \wp{C_1}{G}\\
			&\quad\,+ \alpha \cdot (1 - Pr(P_o)) \cdot \wp{C_2}{F} + (1 - Pr(P_o)) \cdot \wp{C_2}{G}\\
			&= \alpha \cdot (Pr(P_o) \cdot \wp{C_1}{F} + (1 - Pr(P_o)) \cdot \wp{C_2}{F}) \\
			&\quad\,+ Pr(P_o) \cdot \wp{C_1}{G} + (1 - Pr(P_o)) \cdot wp(C_2, G)\\
			&= \alpha \cdot \wp{\IFSYMBOL\dots}{F} + \wp{\IFSYMBOL\dots}{G}
		\end{align*}
		\item \WHILEDO{P_o}{C}
		Note that in \ref{wpConvergence} we showed 
		$$\wp{\WHILEDO{P}{C}}{F} = \lim\limits_{n \rightarrow \infty} \Phi_F^n(0)$$
		where $\Phi_F^n$ is the characteristic function of the while loop with regard to the postexpectation $F$. Furthermore, $\Phi_F$ is calculated using only the weakest preexpectation of the body of the while loop, which is structurally smaller than the loop as a whole and thus the inductive hypothesis can be applied to $\Phi_F$.
		\begin{align*}
			&\wp{\WHILEDO{P}{C}}{\alpha\cdot F + G}\\
			&= \lfp\Phi_{\alpha \cdot F + G} \\
			&= \lim\limits_{n \rightarrow \infty} \Phi_{\alpha \cdot F + G}^n(0)\\
			&=_1 \lim\limits_{n \rightarrow \infty} \alpha \cdot \Phi_{F}^n(0) + \Phi_G^n(0)\\
			&= \alpha \cdot \lim\limits_{n \rightarrow \infty} \Phi_{F}^n(0) + \lim\limits_{n \rightarrow \infty}\Phi_G^n(0)\\
			&= \alpha \cdot \lfp\Phi_{F} + \lfp\Phi_{G}\\
			&= \alpha \cdot \wp{\WHILEDO{P}{C}}F + \wp{\WHILEDO{P}{C}} G
		\end{align*}
		Above, $=_1$ follows as for all $n \in \mathbb{N}$ we have 
		$$\Phi_{\alpha \cdot F + G}^n(0) = \alpha \cdot \Phi_{F}^n(0) + \Phi_G^n(0)$$
		which follows by induction on $n$, the case $n=0$ is trivial. The case $n > 0$ follows by:
		\begin{align*}
			&\Phi_{\alpha \cdot F + G}^n(0)\\
			&= Pr(P_o) \cdot \wp{C} {\Phi_{\alpha \cdot F + G}^{n-1}(0)} + (1-Pr(P_o)) \cdot (\alpha \cdot F + G)\\
			&=_{IH} Pr(P_o) \cdot \wp{C} {\alpha \cdot \Phi_{F}^{n-1}(0) + \Phi_{G}^{n-1}(0)} + (1-Pr(P_o)) \cdot (\alpha \cdot F + G)\\
			&=_{IH} Pr(P_o) \cdot \alpha \cdot \wp{C} {\Phi_{F}^{n-1}(0)} \\
			&\quad\,+ \wp{C} { \Phi_{G}^{n-1}(0)} + (1-Pr(P_o)) \cdot (\alpha \cdot F + G)\\
			&= \alpha \cdot (Pr(P_o) \cdot \wp{C} { \Phi_{F}^{n-1}(0)} + (1-Pr(P_o)) \cdot F) \\
			&\quad\,+ Pr(P_o) \cdot \wp{C} {\Phi_{G}^{n-1}(0)} + (1-Pr(P_o)) \cdot G\\
			&= \alpha \cdot \Phi_{F}^n(0) + \Phi_{G}^n(0)
		\end{align*}
		\item $\COMPOSE{C_1}{C_2}$: The case follows by induction and the definition of \wpsymbol
		\begin{align*}
			&\wp{\COMPOSE{C_1}{C_2}}{\alpha\cdot F + G}\\
			&= \wp{C_1} {\wp{C_2}{\alpha\cdot F + G}} \\
			&=_{IH} \wp{C_1}{\alpha \cdot \wp{C_2}{F} + \wp{C_2} G}\\
			&=_{IH} \alpha \cdot \wp{C_1}{\wp{C_2}{F}} + \wp{C_1}{\wp {C_2} G}\\
			&= \alpha\cdot \wp{\COMPOSE{C_1}{C_2}} F + \wp {\COMPOSE{C_1}{C_2}} G
		\end{align*}
		
		\item \INFER{p(P) \in i}{C_1}{C_2}:
		Fully analogous to the if-statement
		
		\item \CHOOSE{x}{f}:
		\begin{align*}
			&\wp {\CHOOSE{x}{f}} {\alpha\cdot F + G}\, (\beta)\\
			&= (\alpha\cdot F + G)[x \mapsto f] (\beta)\\ 
			&= (\alpha\cdot F + G) (\beta[x \mapsto f])\\ 
			&= \alpha\cdot F \,(\beta[x \mapsto f])+ G\,(\beta[x \mapsto f]) \\
			&= \alpha\cdot \wp {\CHOOSE{x}{f}} {F}\, (\beta) + \wp {\CHOOSE{x}{f}} {G}\,(\beta)
		\end{align*}
		\item \OBSERVEASSIGN{y}{x}:
		The observe statement can be split into two parts: A "pure" observe statement \OBSERVE{x} which only conditions the belief state via $\beta_0|_{x=c}$, and an assignment which updates the value of $y$ with the (now deterministic) value of $x$.
		It suffices to consider the first part, i.e. $\OBSERVE{x}$, as the second part can be proven fully analogous to the case for the assignment, because after the observe statement has been executed there is only a single possible value for $x$ in the belief state.
		Afterwards the full observe statement follows from the proofs for $\OBSERVE{x}$, the assignment and sequential execution.
		\begin{align*}
			&\wp{\OBSERVE{x}}{\alpha\cdot F + G}\\
			&= \sum\limits_{c \in \mathbb{N}} Pr(x=c) \cdot (\alpha\cdot F + G)|_{x=c} \\
			&= \sum\limits_{c \in \mathbb{N}} Pr(x=c) \cdot (\alpha\cdot F|_{x=c} + G|_{x=c}) \\
			&= \sum\limits_{c \in \mathbb{N}} \alpha\cdot Pr(x=c)\cdot F|_{x=c} + Pr(x=c) \cdot G|_{x=c} \\
			&=  \alpha \cdot \sum\limits_{c \in \mathbb{N}} Pr(x=c)\cdot F|_{x=c} + \sum\limits_{c \in \mathbb{N}} Pr(x=c) \cdot G|_{x=c} \\
			&= \alpha \cdot \wp{\OBSERVE{x}}{F} + \wp{\OBSERVE{x}}{G}
		\end{align*}
	\end{itemize}
\end{proof}
\subsection{Proof of \Cref{Unaffected}}\label{UnaffectedProof}
\textit{If none of the variables in proposition $P$ are modified by the loop-free program $C$, then we have 
	\begin{enumerate}
		\item $\wp{C}{Pr(P)} = Pr(P)$
		\item $\wp{C}{[Pr(P) \in i]} = [Pr(P) \in i]$
\end{enumerate}}
\begin{proof}
	We will prove both statements at the same time using induction on $C$. The proof for $[Pr(P) \in i]$ is only shown for $\ASSIGN{x}{E}$, because the cases for the observe and sample statement are fully analogous to the case for $\ASSIGN{x}{E}$, and the other cases for $[Pr(P) \in i]$ are analogous to the proof for $Pr(P)$. 
	\begin{itemize}
		\item \SKIP:
		Follows immediately from the definition of $\wpsymbol$
		
		\item \ASSIGN{x}{E}:
		We have:
		\begin{enumerate}
			\item
			Note that in $=_1$ we use that P is unaffected by $\ASSIGN{x}{E}$, and thus changing the value of $x$ does not affect the result of evaluating $P$. Thus follows $\SMODEL{\rho}{P} = \SMODEL{\rho\statesubst{x}{E}}{P}$. Additionally, $=_2$ holds as there is exactly one $\sigma$ with $\rho[x \mapsto E] = \sigma$ for every $\rho$. Thus every $\rho$ is only factored in exactly once over the whole term.
			\begin{align*}
				\wp{\ASSIGN{x}{E}}{Pr(P)}\,(\beta) &= Pr(P)[x \mapsto E] \,(\beta) \\
				&= Pr(P) \,(\beta[x \mapsto E]) \\
				&= \BMODEL{\beta[x \mapsto E]}{P}\\
				&= \sum\limits_{\sigma \in \Sigma} \beta[x \mapsto E](\sigma) \cdot \SMODEL{\sigma}{P}\\
				&= \sum\limits_{\sigma \in \Sigma} \left(\sum\limits_{\rho \in \Sigma, \rho\statesubst{x}{E} = \sigma} \beta(\rho)\right) \cdot \SMODEL{\sigma}{P}\\
				&= \sum\limits_{\sigma \in \Sigma} \left(\sum\limits_{\rho \in \Sigma, \rho\statesubst{x}{E} = \sigma} \beta(\rho) \cdot \SMODEL{\rho\statesubst{x}{E}}{P}\right)\\
				&=_1 \sum\limits_{\sigma \in \Sigma} \left(\sum\limits_{\rho \in \Sigma, \rho\statesubst{x}{E} = \sigma} \beta(\rho) \cdot \SMODEL{\rho}{P}\right)\\
				&=_2 \sum\limits_{\rho \in \Sigma} \beta(\rho) \cdot \SMODEL{\rho}{P}\\
				&= (\BMODEL{\beta}{P}) = Pr(P)(\beta)
			\end{align*}
			\item \begin{align*}
				\wp{\ASSIGN{x}{E}}{[Pr(P) \in i]}&= [Pr(P) \in i]\statesubst{x}{E} \\
				&= [Pr(P)\statesubst{x}{E} \in i] \\
				&= [\wp{\ASSIGN{x}{E}}{Pr(P)} \in i] \\
				&= [Pr(P) \in i]
			\end{align*}
		\end{enumerate}
		
		\item $\ITE{P_o}{C_1}{C_2}$:
		By definition we have $\mathit{mod}(C_1) \subseteq \mathit{mod}(if\dots)$ and $\mathit{mod}(C_2) \subseteq \mathit{mod}(if\dots)$. Thus $P$ is unaffected by both $C_1$ and $C_2$. With the inductive hypothesis now follows:
			\begin{align*}
				&\wp{\IFSYMBOL\dots} {Pr(P)}\\
				&= Pr(P_o) \cdot \wp{C_1}{Pr(P)} + (1 - Pr(P_o)) \cdot \wp{C_2}{Pr(P)}\\
				&=_{IH} Pr(P_o) \cdot Pr(P) + (1 - Pr(P_o) \cdot Pr(P)\\
				&= (Pr(P_o) + (1 - Pr(P_o)) \cdot Pr(P)\\
				&= Pr(P)
			\end{align*}
		
		\item $\COMPOSE{C_1}{C_2}$:
		By definition we have $\mathit{mod}(C_1) \subseteq \mathit{mod}(\COMPOSE{C_1}{C_2})$ and $\mathit{mod}(C_2) \subseteq \mathit{mod}(\COMPOSE{C_1}{C_2})$. Thus $P$ is unaffected by both $C_1$ and $C_2$. With the inductive hypothesis now follows:
			\begin{align*}
				\wp{\COMPOSE{C_1}{C_2}}{Pr(P)} &= \wp{C_1}{\wp{C_2}{Pr(P)}}\\
				&=_{IH} \wp{C_1}{Pr(P)}\\
				&=_{IH} Pr(P)
			\end{align*}
		
		\item $\INFER{p(P_u)}{C_1}{C_2}$:
		By definition of $\mathit{mod}$ we have $\mathit{mod}(C_1) \subseteq \mathit{mod}(infer\dots)$ and $\mathit{mod}(C_2) \subseteq \mathit{mod}(infer\dots)$. Thus $P$ is unaffected by both $C_1$ and $C_2$. With the inductive hypothesis now follows:
			\begin{align*}
				&\wp{\texttt{infer}\dots}{Pr(P)}\\ 
				&= [Pr(P_u) \in i] \cdot \wp{C_1}{Pr(P)} + [Pr(P_u) \notin i] \cdot \wp{C_2}{Pr(P)}\\
				&=_{IH} [Pr(P_u) \in i] \cdot Pr(P) + [Pr(P_u) \notin i] \cdot Pr(P)\\
				&= ([Pr(P_u) \in i] + [Pr(P_u) \notin i]) \cdot Pr(P)\\
				&= Pr(P)
			\end{align*}
		
		\item $\CHOOSE{x}{f}$:
			$=_1$ follows because $P$ is unaffected by $x$, and thus we have $\SMODEL{\rho}{P} = \SMODEL{\sigma}{P}$. Afterwards, $=_2$ holds because the term $\beta(\rho) \cdot f(\rho, c) \cdot \SMODEL{\rho}{P}$ appears exactly once in the previous expression for every pair $\rho$ and $c$, namely in the summand of the variable assignment $\sigma$ with $\rho[x\mapsto c] = \sigma$. Thus we can reorder the summands in the term and quantify over $c \in \mathbb{N}$ instead of $\sigma \in \Sigma$.
			\begin{align*}
				&\wp{\CHOOSE{x}{f}}{Pr(P)}\,(\beta) \\
				&= Pr(P)\statesubst{x}{f} \,(\beta) \\
				&=  Pr(P) \,(\beta\statesubst{x}{f}) \\
				&= \BMODEL{\beta\statesubst{x}{f}}{P} \\
				&= \sum_{\sigma \in \Sigma} \beta\statesubst{x}{f}(\sigma) \cdot \SMODEL{\sigma}{P}\\
				&= \sum_{\sigma \in \Sigma} \left(\sum\limits_{\rho \in \Sigma, \rho[x\mapsto \sigma(x)] = \sigma} \beta(\rho)\cdot f(\rho, \sigma(x))\right) \cdot \SMODEL{\sigma}{P}\\
				&=_1 \sum_{\sigma \in \Sigma} \left(\sum\limits_{\rho \in \Sigma, \rho[x\mapsto \sigma(x)] = \sigma} \beta(\rho) \cdot f(\rho, \sigma(x)) \cdot \SMODEL{\rho}{P}\right)\\
				&=_2 \sum_{c \in \mathbb{N}} \left(\sum\limits_{\rho \in \Sigma} \beta(\rho) \cdot f(\rho, c) \cdot \SMODEL{\rho}{P}\right)\\
				&= \sum\limits_{\rho \in \Sigma} \beta(\rho) \cdot \SMODEL{\rho}{P} \cdot \sum\limits_{c \in \mathbb{N}} f(\rho, c) \\
				&= \sum\limits_{\rho \in \Sigma} \beta(\rho) \cdot \SMODEL{\rho}{P} \cdot 1 \\
				&= Pr(P)\, (\beta)
			\end{align*}
		
		\item $\OBSERVEASSIGN{y}{x}$:
			The observe statement can be split into two parts: A "pure" observe statement \OBSERVE{x} which only conditions the belief state via $\beta_0|_{x=c}$, and an assignment which updates the value of $y$ with the (now deterministic) value of $x$.
			It suffices to consider the first part, i.e. $\OBSERVE{x}$, as the second part can be proven fully analogous to the case for the assignment, because after the observe statement has been executed there is only a single possible value for $x$ in the belief state.
			Afterwards the full observe statement follows from the proofs for $\OBSERVE{x}$, the assignment and sequential execution.
			We begin by merely rewriting using definitions. 
			\begin{align*}
				\wp{\OBSERVE{x}}{Pr(P)}\,(\beta) &= \sum\limits_{c \in \mathbb{N}} Pr(x = c)\,\beta \cdot Pr(P)|_{x=c} \,(\beta) \\
				&= \sum\limits_{c \in \mathbb{N}} Pr(x = c)\,(\beta) \cdot Pr(P) \,(\beta|_{x=c}) \\
				&= \sum\limits_{c \in \mathbb{N}} Pr(x = c)\,(\beta) \cdot \BMODEL{\beta|_{x=c}}{P} \\
				&= \sum\limits_{c \in \mathbb{N}} Pr(x = c)\,(\beta) \cdot \left(\sum_{\sigma \in \Sigma} \beta|_{x=c}(\sigma) \cdot \SMODEL{\sigma}{P}\right)\\
				&= \sum\limits_{c \in \mathbb{N}} \BMODEL{\beta}{x=c} \cdot \left(\sum_{\sigma \in \Sigma} \beta|_{x=c}(\sigma) \cdot \SMODEL{\sigma}{P}\right)\\
			\end{align*}
			How exactly $\beta|_{x=c}(\sigma)$ is resolved depends on whether or not $\BMODEL{\beta}{x=c} > 0$ holds. Note however that in both cases the following terms are equal :
			\begin{enumerate}
				\item $$\BMODEL{\beta}{x=c} \cdot \left(\sum_{\sigma \in \Sigma} \beta|_{x=c}(\sigma) \cdot \SMODEL{\sigma}{P}\right)$$
				\item$$\BMODEL{\beta}{x=c} \cdot \left(\sum_{\sigma \in \Sigma} \frac{\beta(\sigma)\cdot \SMODEL{\sigma}{x=c}}{\BMODEL{\beta}{x=c}} \cdot \SMODEL{\sigma}{P}\right) $$
			\end{enumerate}
			In the case $\BMODEL{\beta}{x=c} > 0$ this follows directly from the definition of $\beta|_{x=c}(\sigma)$, and in the case of $\BMODEL{\beta}{x=c} = 0$ it holds as we have $\frac{a}{0} = \infty$ and $0 \cdot \infty = 0$ by definition. Furthermore, consider that $\OBSERVE{x}$ affects all unobservable variables, thus we can assume that $P$ is an observable proposition and we have either $Pr(P)\,(\beta) = 1$ or $Pr(P)\,(\beta) = 0$ for every belief state $\beta$. W.l.o.g. we assume $Pr(P)\,(\beta) = 1$, it thus follows for every $\sigma$ with $\beta(\sigma) > 0$ we have $\SMODEL{\sigma}{P} = 1$. We will make use of this in $=_1$. Afterwards, $=_2$ follows as we have $\SMODEL{\sigma}{x=c} = 1$ for exactly one $c\in \mathbb{N}$. We now get:
			\begin{align*}
				&\sum\limits_{c \in \mathbb{N}} \BMODEL{\beta}{x=c} \cdot \left(\sum_{\sigma \in \Sigma} \beta|_{x=c}(\sigma) \cdot \SMODEL{\sigma}{P}\right)	\\
				&= \sum\limits_{c \in \mathbb{N}} \BMODEL{\beta}{x=c} \cdot \left(\sum_{\sigma \in \Sigma} \frac{\beta(\sigma)\cdot \SMODEL{\sigma}{x=c}}{\BMODEL{\beta}{x=c}} \cdot \SMODEL{\sigma}{P}\right)\\
				&= \sum\limits_{c \in \mathbb{N}} \sum_{\sigma \in \Sigma} \beta(\sigma)\cdot \SMODEL{\sigma}{x=c} \cdot \SMODEL{\sigma}{P}\\
				&=_1 \sum\limits_{c \in \mathbb{N}} \sum_{\sigma \in \Sigma, \beta(\sigma) > 0} \beta(\sigma)\cdot \SMODEL{\sigma}{x=c} \cdot 1\\
				&=_2 \sum_{\sigma \in \Sigma, \beta(\sigma) > 0} \beta(\sigma) \cdot 1\\
				&= 1 = Pr(P)\, (\beta)
			\end{align*}
	\end{itemize}
\end{proof}

\section{Proofs and Derivations for Syntactic Predicates}\label{GrammarProofs}

\subsection{Proof for Theorem~\ref{GrammarNew}}\label{GrammarProofNew}
We consider the following grammar, as defined in \Cref{GrammarNew}
\begin{align*}
	E &\Coloneqq x \in \Varso \cup \Varsu \,|\, c \in \mathbb{N} \,|\, [P_u], \,|\, E \oplus E & \oplus \in \{+, -, \cdot, /\}\\
	\Phi &\Coloneqq Ex(E) \sim q \cdot Ex(E) & \sim \in \{<,\leq, =, \neq, \geq, >\} \\
	G &\Coloneqq Ex(E) \,|\, G + G \,|\, \alpha \cdot G \,|\, [\Phi] \cdot G 
\end{align*}
$$Ex(E) \Coloneqq \lambda\, \beta \mydot \sum_{\sigma \in \Sigma} \beta(\sigma)\cdot \sigma(E)$$
To show: For any loop-free free program $C$ and predicate $F \in G$ we have $\wp{C}{F} \in G$\\

Proof by structural induction over $C$:
\begin{itemize}
	\item $\SKIP$: We have $\wp{\SKIP}{F} = F \in G$ by assumption\\
	\item $\ASSIGN{x}{E}$: We will prove that $\wp{\ASSIGN{x}{E}} F = F\statesubst{x}{E} \in G$ holds using structural induction on $F$ :
	\begin{itemize}
		\item $F = Ex(E')$: By theorem \Cref{lem:replacement} it suffices to show that $E'[x/E]$ is an extended expression, which follows from the definition of extended expressions
		\item $F = F_1 + F_2$: By induction we have $F_1\statesubst{x}{E} \in G$ and $F_2\statesubst{x}{E} \in G$. It now follows that we have 
		$$F\statesubst{x}{E} = (F_1 + F_2)[x \mapsto E] = F_1\statesubst{x}{E} + F_1\statesubst{x}{E} \in G$$
		\item $F = \alpha \cdot F'$ follows by the inductive hypothesis for $F'$
		\item $F = [Ex(E) \sim q\cdot Ex(E)]\cdot F'$ Note that we have 
		\begin{align*}
			&\Bigl(\bigl(\bigl[Ex(E_1) \sim q\cdot Ex(E_2)\bigr]\cdot F'\bigr)\statesubst{x}{E}\Bigr)(\beta)\\ 
			&= \Bigl(\bigl[Ex(E_1) \sim q\cdot Ex(E_2)\bigr]\cdot F'\Bigr)(\beta\statesubst{x}{E})\\ 
			&= \bigl[Ex(E_1)(\beta\statesubst{x}{E}) \sim q\cdot Ex(E_2)(\beta\statesubst{x}{E})\bigr]\cdot F'(\beta\statesubst{x}{E})\\
			&= \bigl[Ex(E_1)\statesubst{x}{E}(\beta) \sim q\cdot Ex(E_2)\statesubst{x}{E}(\beta)\bigr]\cdot F'\statesubst{x}{E}(\beta)\\
			&= \Bigl(\bigl[Ex(E_1)\statesubst{x}{E} \sim q\cdot Ex(E_2)\statesubst{x}{E}\bigr]\cdot F'\statesubst{x}{E}\Bigr)(\beta)
		\end{align*}
		By the inductive hypothesis we have $F'\statesubst{x}{E} \in G$, it now remains to show that $[Ex(E_1)\statesubst{x}{E} \sim q\cdot Ex(E_2)\statesubst{x}{E}] \cdot \dots$ can be expressed within $G$. This is the case by \Cref{lem:replacement}, as $Ex(E_i)\statesubst{x}{E} = Ex(E_i[x/E])$ holds and $E_i[x/E]$ is still an extended expression.
	\end{itemize}
	
	\item $\ITE{P_o}{C_1}{C_2}$: We have 
	$$\wp{\IFSYMBOL\dots} F = Pr(P_o) \cdot \wp{C_1}{F} + (1-Pr(P_o)) \cdot \wp{C_2}{F}$$
	By the inductive hypothesis, we can conclude that both $\wp{C_1}{F}\in G$ and $\wp{C_2}{F}\in G$ hold. It now suffices to show that $Pr(P_o) \cdot \dots$ can be expressed within $G$.
	
	Note that $P_o$ is an observable proposition, and thus we have either $Pr(P_o) = 1$ or $Pr(P_o) = 0$, which means we have $Pr(P_o) = [Pr(P_o) = 1] = [Ex([P_o]) = 1]$, which is expressible within $G$ and thus $\wp{\IFSYMBOL\dots} F\in G$ holds. Similarly, we have $1- Pr(P_o) = [Ex([P_o]) = 0]$
	
	\item $\COMPOSE{C_1}{C_2}$: By the inductive hypothesis for $C_2$, we have $\wp{C_2}{F} \in G$, and with the inductive hypothesis for $C_1$ we get 
	$$\wp{\COMPOSE{C_1}{C_2}} F = \wp{C_1}{\wp{C_2}{F}} \in G$$
	
	\item $\INFER{p(P_u)\in i}{C_1}{C_2}$: We have
	$$\wp{\texttt{infer}\dots}{F} = [Pr(P_u) \in i] \cdot \wp{C_1}{F} + [Pr(P_u) \notin i] \cdot \wp{C_2}{F}$$
	By the inductive hypothesis, we can conclude that both $\wp{C_1}{F}\in G$ and $\wp{C_2}{F}\in G$ hold. It now suffices to show that $[Pr(P_u) \in i]\cdot \dots$ can be expressed within $G$.
	
	Note that we restricted infer statements to only include conditions which can be expressed with a comparison operator, which means that we have $[Pr(P_u) \in i] = [Pr(P_u) \sim q]$ for some $q$ and $\sim \in \{<,\leq, =, \neq, \geq, >\}$.
	Furthermore, we have $[Pr(P_u) \sim q] = [Ex([P_u]) \sim q] = [Ex([P_u]) \sim q \cdot Ex(1)]$ by the definition of $Ex$.  Thus the Iverson Bracket as generated by the weakest preexpectation of the infer statement can be expressed within $G$, and $\wp{\texttt{infer}\dots}{F}\in G$ holds.
	
	\item $\CHOOSE{x}{f}$: We will prove $\wp{\CHOOSE{x}{f}}{F} = F[x\mapsto f] \in G$ by structural induction on $F$. However, most of the cases and reasoning are fully analogous to the previously shown case for $\ASSIGN{x}{E}$, and thus we will only prove that $Ex(E)[x\mapsto f] \in G$ holds. All other steps of the proof are exactly the same as they were for $\ASSIGN{x}{E}$.

	Note that by assumption we can write $f$ as a finite sum $\sum_{0\leq i < n} p_i |E_i\rangle$ for some fixed expressions $E_i$ and probabilities $p_i$, so that we have $f(\sigma) = \sum_{0\leq i < n} p_i |\sigma(E_i)\rangle$. 
	By \Cref{BeliefElimTheoremNew} we have $Ex(E)\statesubst{x}{f} = Ex(E\leftarrow f)$, and $E\leftarrow f$ is an extended expression:
	$$E \leftarrow f = \sum_{i \in \mathbb{N}} p_i \cdot E[x / E_i]$$
	\item $\OBSERVE{x}$:
	The observe statement can be split into two parts: A "pure" observe statement \OBSERVE{x} which only conditions the belief state via $\beta_0|_{x=c}$, and an assignment which updates the value of $y$ with the (now deterministic) value of $x$.
	It suffices to consider the first part, i.e. $\OBSERVE{x}$, as the second part can be proven fully analogous to the case for the assignment, because after the observe statement has been executed there is only a single possible value for $x$ in the belief state.
	Afterwards the full observe statement follows from the proofs for $\OBSERVE{x}$, the assignment and sequential execution.
	
	We first consider the structure of the predicate $F$. In particular, we will exhaustively apply distributivity and commutativity so that the predicate has the form 
	$$F = \sum\limits_{0\leq i < n} \left(\prod\limits_{0\leq j < m} \alpha_j\right) \cdot \left(\prod\limits_{0\leq k < l} [\Phi_k]\right) \cdot Ex(E_i) $$
	This allows us to apply \Cref{Linearity} to split the predicate up, so that each individual sub-predicate contains only a product of guards in the form of Iverson Brackets, followed by a single $Ex(E)$:
	\begin{align*}
		&\wp{\OBSERVE{x}}{F}\\
		&= \wp{\OBSERVE{x}}{ \sum\limits_{0\leq i < n} \left(\prod\limits_{0\leq j < m} \alpha_j\right) \cdot \left(\prod\limits_{0\leq k < l} [\Phi_k]\right) \cdot Ex(E_i)}\\ 
		&=  \sum\limits_{0\leq i < n} \left(\prod\limits_{0\leq j < m} \alpha_j\right) \cdot \wp{\OBSERVE{x}}{\left(\prod\limits_{0\leq k < l} [\Phi_k]\right) \cdot Ex(E_i)}
	\end{align*}
	It now suffices to show $\wp{\OBSERVE{x}} {(\prod_{0\leq k < l} [\Phi_k]) \cdot Ex(E_i)} \in G$. By the definition of $\wpsymbol$ we get:
	\begin{align*}
		&\wp{\OBSERVE{x}} {\left(\prod_{0\leq k < l} [\Phi_k]\right) \cdot Ex(E_i)} \\
		&= \sum\limits_{c \in \mathbb{N}} Pr(x = c) \cdot \left(\prod_{0\leq k < l} [\Phi_k]|_{x=c}\right) \cdot Ex(E_i)|_{x=c}
	\end{align*}
	We will now consider one term of the sum for a fixed $c$, and consider how the term can be simplified to be expressed in $G$. We will apply \Cref{BeliefConditioningTheoremNew} in order to resolve the term $Ex(E_i)|_{x=c}$, which allows us to do the following:
	\begin{align*}
		&Pr(x = c) \cdot \left(\prod_{0\leq k < l} [\Phi_k]|_{x=c}\right) \cdot Ex(E_i)|_{x=c}\\
		&= Ex([x = c]) \cdot \left(\prod_{0\leq k < l} [\Phi_k]|_{x=c}\right) \cdot \frac{Ex(E_i\cdot [x=c])}{Ex([x=c])}\\
		&= \left(\prod_{0\leq k < l} [\Phi_k]|_{x=c}\right) \cdot Ex(E_i\cdot [x=c])\\
	\end{align*}
	It now only remains to show that $[\Phi_k]|_{x=c}$ is a valid guard within $G$, which is shown as follows by again applying \Cref{BeliefConditioningTheoremNew}:
	\begin{align*}
		[\Phi_k]|_{x=c} &= [Ex(E_1) \sim q \cdot Ex(E_2)]|_{x=c} \\
		&= [Ex(E_1)|_{x=c} \sim q \cdot Ex(E_2)|_{x=c}]\\
		&= \left[\frac{Ex(E_1\cdot [x=c])}{Ex([x=c])} \sim q \cdot \frac{Ex(E_2\cdot [x=c])}{Ex([x=c])}\right]\\
		&= [Ex(E_1\cdot [x=c]) \sim q \cdot Ex(E_2\cdot [x=c])]
	\end{align*}
	Here we made use of the fact that $\Phi_k$ was a valid guard within $G$ beforehand, which allowed us to rewrite the condition. Finally, canceling out the denominators on both sides of $\sim$ is a valid transformation for all the allowed operators in the set $\{<,\leq, =, \neq, \geq, >\}$.
	
\end{itemize}

\subsection{Proof for Theorem~\ref{BeliefConditioningTheoremNew}}\label{ConditioningProofNew}
To show: $Ex(E)|_{x = c} = \frac{Ex(E\cdot [x=c])}{Ex([x = c])}$.
We make a case distinction on $Ex([x=c])$: 
The case $Ex([x=c]) = 0$ follows as $Ex([x=c]) = 0$ implies $Pr(x=c) = 0$, which by definition of $\beta|_{x=c}$ yields $Ex(E)|_{x = c} = 0$.
Afterwards, $\frac{Ex(E\cdot [x=c])}{Ex([x = c])}$ follows by our definition for division by $0$ always yielding $0$.
If we have  for $Ex([x=c]) > 0$, then the theorem can be shown as follows:
\begin{align*}
	Ex(E)|_{x = c} 
	&= \sum_{\sigma \in \Sigma} \beta|_{x=c}(\sigma) \cdot \sigma(E) \\
	&= \sum_{\sigma \in \Sigma} \frac{\beta(\sigma) \cdot \SMODEL{\sigma}{x=c}}{\BMODEL{\beta}{x=c}} \cdot \sigma(E) \\
	&= \frac{1}{\BMODEL{\beta}{x=c}} \cdot \sum_{\sigma \in \Sigma} \beta(\sigma) \cdot \SMODEL{\sigma}{x=c} \cdot \sigma(E) \\
	&= \frac{1}{Pr(x=c)} \cdot \sum_{\sigma \in \Sigma} \beta(\sigma) \cdot \sigma(E \cdot [x=c]) \\
	&= \frac{1}{Ex([x=c])} \cdot Ex(E \cdot[ x=c]) \\
	&= \frac{Ex(E\cdot [x=c])}{Ex([x = c])}
\end{align*}
\subsection{Proof for Theorem~\ref{BeliefElimTheoremNew}}\label{BeliefElimProofNew}
To show: $Ex(E)\statesubst{x}{f} =  
Ex(E\leftarrow f)$ 
with $f= \sum_{0< i \leq n} p_i|E_i\rangle$\\
For intuition, keep in mind that in the proof below $\sigma$ is a variable assignment after assigning a new value to $x$ using the function $f$ and $\rho$ is a variable assignment before applying the function $f$. We now get:
\begin{align*}
	Ex(E)\statesubst{x}{f} &= \sum_{\sigma \in \Sigma} \beta\statesubst{x}{f}(\sigma) \cdot \sigma(E) \\
	&= \sum_{\sigma \in \Sigma} \left(\sum\limits_{\rho \in \Sigma, \rho[x\mapsto \sigma(x)] = \sigma} \beta(\rho) \cdot f(\rho, \sigma(x))\right) \cdot \sigma(E) \\
	&= \sum_{\sigma \in \Sigma} \sum\limits_{\rho \in \Sigma, \rho[x\mapsto \sigma(x)] = \sigma} \left( \beta(\rho) \cdot f(\rho, \sigma(x)) \cdot \sigma(E)\right)\\
	&=_1 \sum_{\rho \in \Sigma} \sum\limits_{\sigma \in \Sigma, \rho[x\mapsto \sigma(x)] = \sigma} \left( \beta(\rho) \cdot f(\rho, \sigma(x)) \cdot \sigma(E)\right)\\
	&= \sum_{\rho \in \Sigma} \beta(\rho) \cdot \sum\limits_{\sigma \in \Sigma, \rho[x\mapsto \sigma(x)] = \sigma} f(\rho, \sigma(x)) \cdot \rho[x\mapsto \sigma(x)](E)\\
	&= \sum_{\rho \in \Sigma} \beta(\rho) \cdot  \sum\limits_{\sigma \in \Sigma, \rho[x\mapsto \sigma(x)] = \sigma} f(\rho, \sigma(x)) \cdot \rho(E[x/\sigma(x)])\\
	&=_2 \sum_{\rho \in \Sigma} \beta(\rho) \cdot \sum_{c \in \mathbb{N}} f(\rho, c) \cdot \rho(E[x / c])\\
	&=_3 \sum_{\rho \in \Sigma} \beta(\rho) \cdot \sum_{0< i \leq n} p_i \cdot \rho(E[x / E_i])\\
	&= \sum_{\rho \in \Sigma} \beta(\rho) \cdot \rho(E\leftarrow f)\\
	&= Ex(E\leftarrow f)
\end{align*}
For $=_1$ we simply switch first summing up over all $\rho$ instead of summing over all $\sigma$, note that this does not affect the total summands in the expression, but merely re-orders them: If we consider two variable assignments $\sigma'$ and $\rho'$, then the term $ \beta(\rho') \cdot f(\rho', \sigma'(x)) \cdot \sigma'(E) $ appears exactly once in both the sum before reordering and in the sum after reordering if $\rho'[x\mapsto \sigma'(x)] = \sigma'$ holds, and does not appear in either sum otherwise.

For $=_2$ one should notice that the first sum essentially only considers variable assignments of the form $\sigma = \rho[x \mapsto c]$, as only those variable assignments fulfill the condition $\rho[x\mapsto \sigma(x)] = \sigma$. It is thus possible to quantify the sum over values instead of variable assignments.

For $=_3$, we use that we have $f= \sum_{0< i \leq n} p_i|E_i\rangle$, thus rather than summing over the probability to sample a particular value $c \in \Nats$ we can also consider the probability to sample an Expression $E_i$.

\subsection{Proof of the Loop Invariant}\label{wpCalc}
We would like to prove that the Predicate 
$$ I = (1 - Pr(\mathit{inCare})) \cdot Pr(d=1) + Pr(\mathit{inCare}) \cdot q $$
is a loop Invariant with regard to the program in \Cref{CaseStudyLoopy} with Post Expectation $F= Pr(d=1)$, thus we need to show:
$$ I \geq \Phi_F(I) = (1 - Pr(\mathit{inCare})) \cdot F + Pr(\mathit{inCare}) \cdot \wp{C}{I}$$ 
The proof is done as follows:
\begin{align*}
	\Phi_F(I) &= (1 - Pr(\mathit{inCare})) \cdot F + Pr(\mathit{inCare}) \cdot \wp{\texttt{infer}\dots}{I}\\
	&= (1 - Pr(\mathit{inCare})) \cdot Pr(d=1) + Pr(\mathit{inCare}) \cdot \wp{\texttt{infer}\dots}{I}\\
	&= (1 - Pr(\mathit{inCare})) \cdot Pr(d=1)\\ 
	&\qquad+ Pr(\mathit{inCare}) \cdot ([Pr(d=1) > q] \cdot \wp{\texttt{sample}\dots}{I}\\ 
	&\qquad \qquad \qquad \qquad+ [Pr(d=1) \leq q]\cdot \wp{\mathit{inCare}=false} I )\\
	&=_a (1 - Pr(\mathit{inCare})) \cdot Pr(d=1)\\ 
	&\qquad+ Pr(\mathit{inCare}) \cdot ([Pr(d=1) > q] \cdot q\\ 
	&\qquad \qquad \qquad \qquad+ [Pr(d=1) \leq q]\cdot \wp{\mathit{inCare}=false} I )\\
	&\leq_b (1 - Pr(\mathit{inCare})) \cdot Pr(d=1) \\
	&\qquad+ Pr(\mathit{inCare}) \cdot ([Pr(d=1) > q] \cdot q + [Pr(d=1) \leq q]\cdot q)\\
	&= (1 - Pr(\mathit{inCare})) \cdot Pr(d=1) + Pr(\mathit{inCare}) \cdot q\\
	&= I
\end{align*}
Sub proof for $=_a$:
\begin{align*}
	&\quad \;\,Pr(\mathit{inCare}) \cdot [Pr(d=1) > q] \cdot \wp{\texttt{sample}\dots}{I}\\ 
	&=_1 Pr(\mathit{inCare}) \cdot [Pr(d=1) > q] \cdot (\wp{l.5-l.7} {1 - Pr(\mathit{inCare}) \cdot Pr(d=1)}\\
	&\qquad\qquad\qquad\qquad\qquad\qquad\quad \;\; + \wp{l.5-l.7} {Pr(\mathit{inCare})} \cdot q)\\
	&=_2 Pr(\mathit{inCare}) \cdot [Pr(d=1) > q] \cdot (\wp{l.5-l.7} {(1 - Pr(\mathit{inCare})) \cdot Pr(d=1)} \\
	&\qquad\qquad\qquad\qquad\qquad\qquad\quad \;\; + Pr(\mathit{inCare}) \cdot q\large)\\
	&=_3 Pr(\mathit{inCare}) \cdot [Pr(d=1) > q] \cdot (0 \cdot \wp{l.5-l.7}  {Pr(d=1)} + 1 \cdot q)\\
	&=\, \,Pr(\mathit{inCare}) \cdot [Pr(d=1) > q] \cdot q
\end{align*}
$=_1$ follows from both \Cref{Linearity}, $=_2$ follows from \Cref{Unaffected}. Note that $\mathit{inCare}$ is an observable variable, and thus $Pr(\mathit{inCare})$ will be $0$ or $1$, if $Pr(\mathit{inCare}) = 0$, then the term as a whole will be $0$. 

For $=_3$ we consider the case with $Pr(\mathit{inCare}) = 1$ and one can show that in $\wp{l.5-l.7} {(1 - Pr(\mathit{inCare}) )\cdot Pr(d=1)}$  the expression $(1 - Pr(\mathit{inCare})$ will be unaffected by $\wpsymbol$ and is always $0$. 
This follows syntactically when we apply the rules of $\wpsymbol$ to determine the value of $\wp{l.5-l.7} {(1 - Pr(\mathit{inCare})) \cdot Pr(d=1)}$, but it can be understood as a more general case of \Cref{Unaffected}.
\\

Sub proof for $\leq_b$:
\begin{align*}
	& \;Pr(\mathit{inCare}) \cdot [Pr(d=1) \leq q] \cdot \wp{\mathit{inCare} = false} I\\
	&= Pr(\mathit{inCare}) \cdot [Pr(d=1) \leq q] \\
	&\qquad \cdot \left((1 - Pr(\mathit{inCare})) \cdot Pr(d=1) + Pr(\mathit{inCare}) \cdot q\right)[\mathit{inCare} \mapsto false]\\ 
	&= Pr(\mathit{inCare}) \cdot [Pr(d=1) \leq q] \cdot ((1 - Pr(\mathit{false}) \cdot Pr(d=1) + Pr(\mathit{false}) \cdot q)\\
	&= Pr(\mathit{inCare}) \cdot [Pr(d=1) \leq q] \cdot Pr(d=1) \\
	&\leq Pr(\mathit{inCare}) \cdot [Pr(d=1) \leq q] \cdot q
\end{align*}
The last inequality holds as we have $[Pr(d=1) \leq q] = 0$ unless we have $Pr(d=1) \leq q$, thus $q$ is an upper bound on $Pr(d=1) $ whenever the expression as a whole is non-zero
\end{document}